\newif\iffinal
\newif\ifarxiv
\finaltrue 
\arxivtrue 

\documentclass[a4paper,USenglish,cleveref, autoref, thm-restate, pdfa]{lipics-v2021}

\iffinal
\usepackage[colorinlistoftodos,prependcaption,textsize=tiny,disable]{todonotes}
\else
\usepackage[colorinlistoftodos,prependcaption,textsize=tiny]{todonotes}
\fi
\newcommand{\todomi}[1]{\todo[inline]{#1}}

\usepackage{comment}
\usepackage{amsmath}
\usepackage{amsthm}
\usepackage{amssymb}
\usepackage{mathabx}
\usepackage{xspace}
\usepackage{latexcolors} 
\usetikzlibrary{arrows}
\usepackage{enumitem} 
\setlist[itemize]{label=\textbullet} 

\NewDocumentCommand\set{sm}{\IfBooleanTF#1{\{{#2}\}}{\left\{{#2}\right\}}}
\NewDocumentCommand\ceil{sm}{\IfBooleanTF#1{\lceil{#2}\rceil}{\left\lceil{#2}\right\rceil}}
\NewDocumentCommand\floor{sm}{\IfBooleanTF#1{\lfloor{#2}\rfloor}{\left\lfloor{#2}\right\rfloor}}
\NewDocumentCommand\pare{sm}{\IfBooleanTF#1{({#2})}{\left({#2}\right)}}
\NewDocumentCommand\range{smm}{\IfBooleanTF#1{\set*{{#2},\dots,{#3}}}{\set{{#2},\dots,{#3}}}}
\NewDocumentCommand\card{sm}{\IfBooleanTF#1{|{#2}|}{\left|{#2}\right|}}

\def\leq{\leqslant}\def\le{\leq}
\def\geq{\geqslant}\def\ge{\geq}
\def\emptyset{\varnothing}

\def\eps{\varepsilon}

\def\FTP{\textsc{FTP}\xspace}
\def\FTRP{\textsc{FTRP}\xspace}

\def\id{\textsc{id}} 
\def\ano{\textsc{a}} 
\def\FTRPid{\textsc{FTRP}\id\xspace}
\def\FTRPa{\textsc{FTRP}\ano\xspace}

\def\OPT{\ensuremath{\gamma}\xspace}
\def\OPTftp{\overline{\OPT}} 
\def\OPTid{\overline{\OPT_\id\xspace}} 
\def\OPTa{\overline{\OPT_\ano\xspace}} 
\def\OPTidn{\OPT_\id\xspace} 
\def\OPTan{\OPT_\ano\xspace} 

\def\pos{P} 
\def\P{\mathcal{P}} 
\def\R{\mathcal{R}} 
\def\wutree{\textsf{T}\xspace} 
\def\wutime{\ensuremath{\text{wu}}\xspace} 
\def\returntime{\ensuremath{\text{co}}}
\newcommand{\completionTime}{completion time}

\def\returnarc{\ensuremath{\mathcal{E}_\textsc{r}}\xspace}
\def\wakeuparc{\ensuremath{{\mathcal{E}_\wutree}\xspace}}

\def\last{\textsf{last}}
\def\chord{\textsf{chord}}
\def\gchord{\textsf{gchord}}

\def\best{\textsf{A}}
\def\Rbest{{r^{\bigstar}}} 
\def\worst{\textsf{W}}
\def\Rworst{r_\worst}
\def\good{\textsf{G}}
\def\bad{\textsf{B}} 

\def\Bad{\mathcal B} 
\def\Lset{\mathcal{L}} 

\def\ALGFTR{\textsf{FTR-greedy}}
\def\linear{\textsc{Linear}\xspace}
\def\Tlinear{\ensuremath{t_{\textsc{l}}}} 
\def\convex{\textsc{Convex}\xspace}
\def\crown{\textsc{Crown}\xspace}
\def\cone{\textsc{Cone}\xspace}
\def\diam{\textsf{diam}}

\def\splitconestrategy{\textsc{SplitConeStrategy}\xspace}

\def\FigExCov{1}
\def\FigExId{2}
\def\FigExIdBi{3}
\def\FigRegularOptBiIdThreeA{4}
\def\FigRegularOptBiIdFourA{5}
\def\FigLargeNEmpty{6}
\def\FigLargeN{7}
\def\FigLargeNCaseOne{8}
\def\FigLargeNCaseTwo{9}
\def\FigRegularOptBiIdThree{10}
\def\FigRegularOptCovThree{11}
\def\FigRegularOptBiIdFour{12}
\def\FigRegularOptCovFour{13}
\def\FigRegularOptBiIdFive{14}
\def\FigRegularOptCovFive{15}
\def\FigRegularOptBiIdSix{16}
\def\FigRegularOptCovSix{17}
\def\FigRegularOptBiIdSeven{18}
\def\FigRegularOptCovSeven{19}
\def\FigRegularOptBiIdEight{20}
\def\FigRegularOptCovEight{21}
\def\FigWorstIdCA{22}
\def\FigWorstIdCB{23}
\def\FigWorstIdCC{24}
\def\FigWorstIdCD{25}
\def\FigWorstCovCA{26}
\def\FigWorstCovCB{27}
\def\FigWorstCovCC{28}
\def\FigWorstCovCD{29}
\def\FigWorstIdA{30}

\def\FigWorstIdD{33}
\def\FigWorstCovA{34}

\def\FigWorstCovD{37}

\newcommand{\incfig}[2]{%
    \includegraphics[page=#1, width=\textwidth*#2]{fig/geom_figs.pdf}
}

\title{Freeze-Tag with Return}

\ifarxiv
\date{\today}
\fi

\author{Nicolas Bonichon}%
{LaBRI, University of Bordeaux, CNRS, Bordeaux INP, France}%
{bonichon@labri.fr}%
{https://orcid.org/0000-0002-7012-6851}{}

\author{Cyril Gavoille}%
{LaBRI, University of Bordeaux, CNRS, Bordeaux INP, France}%
{gavoille@labri.fr}%
{https://orcid.org/0000-0003-3671-8607}{}

\author{Nicolas Hanusse}%
{LaBRI, University of Bordeaux, CNRS, Bordeaux INP, France}%
{hanusse@labri.fr}%
{https://orcid.org/0009-0008-9082-7437}{}

\author{Gabriel Le~Bouder}%
{LIP6, Sorbonne University, France}%
{gabriel.le-bouder@lip6.fr}%
{https://orcid.org/0009-0001-2479-5405}{}

\author{Taïssir Marcé}%
{LaBRI, University of Bordeaux, CNRS, Bordeaux INP, France}%
{taissir.marce@labri.fr}%
{https://orcid.org/0009-0001-6836-6311}{}

\author{Nils Morawietz}{LaBRI, Université de Bordeaux, France\\Friedrich Schiller University Jena, Institute of Computer Science, Germany}{nils.morawietz@uni-jena.de}{https://orcid.org/0000-0002-7283-4982}{}

\funding{The authors are supported by the French ANR, project ANR-22-CE48-0001 (TEMPOGRAL)} 

\Copyright{Nicolas Bonichon, Cyril Gavoille, Nicolas Hanusse, Gabriel Le~Bouder, Taïssir Marcé, and Nils Morawietz}

\authorrunning{N. Bonichon, C. Gavoille, N. Hanusse, G. Le~Bouder, T. Marcé, and N. Morawietz}



\ifarxiv
\pdfoutput=1 
\hideLIPIcs  
\fi

\nolinenumbers

\iffinal
\long\def\nicolasB#1{}
\long\def\cyril#1{}
\long\def\nicolasH#1{}
\long\def\nicolasHmargin#1{}
\long\def\taissir#1{}
\long\def\taissirmargin#1{}
\long\def\gabriel#1{}
\long\def\gabrielraw#1{}
\long\def\nils#1{}
\else
\long\def\nicolasB#1{\textcolor{olive}{[Nicolas B: {#1}]}}
\long\def\cyril#1{\textcolor{blue}{[Cyril: {#1}]}}
\long\def\nicolasH#1{\textcolor{brown}{Nicolas H: #1}}
\long\def\nicolasHmargin#1{\todo[color=brown!40]{Nicolas H: #1}}
\long\def\gabrielraw#1{\textcolor{red}{[Gabriel: #1]}}
\long\def\gabriel#1{\todo[color=red!40]{Gabriel: #1}}
\long\def\taissirmargin#1{\todo[color=teal!40]{Taïssir: #1}}
\long\def\taissir#1{\textcolor{teal}{[Taïssir: {#1}]}}
\long\def\nils#1{\textcolor{purple}{[Nils: {#1}]}}
\fi

\ifarxiv
\usepackage{multicol}
\fi

\begin{document}

\maketitle

\begin{abstract}
    In the standard Freeze-Tag Problem (FTP), an initially awake robot (the \emph{source}) is in charge of waking up a swarm of sleeping robots by moving towards them, given that all the awake robots can participate in the awakening process. The goal is to minimize the makespan to wake up all robots assuming they move at unit speed. In this paper we introduce the \emph{Freeze-Tag-with-Return Problem} (FTRP) variant, where the robots must eventually return to their initial positions.

    In the Euclidean plane with $n$ sleeping robots lying on the unit disk centered at the initial position of the source, we show a non-trivial relationship between \FTP and \FTRP by proving that the difference between the optimal makespan of both problems never exceeds $1.959$, and is at least $1.732$ in the worst-case.
    We also present several upper and lower bounds on the optimal makespan.
    In particular, we show that if the sleeping robots are in convex positions, then the optimal makespan is at most $2 + 2\sqrt{2}$, which is achieved by some instances.

    From an algorithmic point-of-view, we present single-exponential algorithms for general distance functions.
    In metric spaces, these algorithms are asymptotically optimal under the ETH, which we show via an NP-hardness reduction on unweighted graphs.
\end{abstract}

\ifarxiv
\newpage
  \tableofcontents
\newpage
\fi

\setcounter{page}{1}
\section{Introduction}%
\label{sec:intro}

The Freeze-Tag Problem (\FTP) was introduced by~\cite{ABFMS02} and describes the following problem: given a set of $n$ \emph{sleeping} robots and \emph{one initial awake} robot, the goal is to wake up all robots as fast as possible. In the context of energy-constrained autonomous robotics, the sleep state provides the necessary downtime for energy harvesting processes.
What makes the problem interesting is that an awake robot must move to a sleeping robot to wake it up. Furthermore, once several robots are awake, they can collaborate to wake up the remaining sleeping robots. The \emph{makespan} is defined as the time required to wake up all robots.  Robots are assumed to move at a constant speed and thus the time can also be expressed as the total length of a longest trajectory followed by a sequence of awake robots to wake up a sleeping robot.

In the following, we propose a new version of this problem, in which the robots, once awakened and after participating in the wake-up process, must return to their initial positions. The objective is now to minimize the time, still referred to as the \emph{makespan}, required for all robots to be both awake and back to an initial position. We call this problem the Freeze-Tag-with-Return Problem (\FTRP). We introduce two variants: in the first, called \FTRPid, each robot has an \emph{identifier} and must return exactly to its own initial position, while in the second, called \FTRPa, the robots are \emph{anonymous} and only need to end at some arbitrary initial position, provided that each initial position contains a robot at the end.

Our problem is motivated by scenarios where robots, once awakened and their wake-up task is completed, must come back to perform specific tasks at their initial positions: inspection of infrastructures or environment (fire detection), cleaning up an area, collecting or delivering objects like batteries, or perform mechanical maintenance. When robots are interchangeable, the $\FTRPa$ variant is enough, whereas when robots have specific roles or skills, the $\FTRPid$ variant is necessary.

Note that for any configuration of initial positions, the optimal makespan for the \FTRP (both variants) is at most the makespan of the \FTP plus the largest distance $D$ between any two initial positions (to allow each robot to come back once awakened), and is at least the optimal makespan of the \FTP plus the smallest distance between any two initial positions (since the robot awakening the last sleeping robot needs to go to another initial position).
For the graph setting, the \emph{overhead due to returns} can reach $D$. Take an unweighted star of $3$ leaves containing sleeping robots and the initial awake robot located at the center. The makespan of \FTP and \FTRP for this instance are respectively $3$ and $5$.
In this work we consider the question of whether these bounds for the overhead are tight.

\subparagraph{State of the art.}

The $\FTP$ has received significant attention since its introduction~\cite{ABFMS02}. Several complexity results have been obtained, showing that the problem is NP-hard, first for metrics induced by weighted star graphs and unweighted trees~\cite{ABFMS06}, then for Euclidean metrics in dimension~3~\cite{Johnson17}, $\ell_p$ metrics in dimension 3 for $p>1$~\cite{DR17}, for the Euclidean metric in dimension 2~\cite{AAY17}, for $\ell_1$ metric in dimension~3~\cite{PdOS23}, and recently claimed for $\ell_1$ metric in dimension~2~\cite{PdOS25}. On the other hand, there is a known PTAS for \FTP in $(\mathbb{R}^d, \ell_p)$ that runs in time $O(n\log n)+ 2^{(d/\eps)^{O(d)}}$~\cite{ABFMS06}, subject to $\eps \leq \eps_d$, where $\eps_d$ depends on the dimension~$d$.

A key parameter for an FTP instance is its \emph{radius}, the maximum distance between the initial robot and any sleeping robot. Obviously, the radius is a lower bound for the optimal makespan. In the Euclidean plane, the worst-case optimal makespan for an instance of radius~1 with $n$ sleeping robots is denoted by $\OPT(n)$.
To warm up, we can observe that $\OPT(1) = 1$ and $\OPT(2) = \OPT(3) = 3$ (take sleeping robots being  diametrically opposed). We call the \emph{wake-up constant} $\OPTftp = \sup_{n \in \mathbb{N}} \OPT(n)$.
As the authors of~\cite{BCGH24} proved, $\OPTftp \ge \OPT(4) = 1+2\sqrt{2} \approx 3.83$, and they have conjectured that this value is the worst possible, that is, $\OPTftp = \OPT(4)$. Actually they have proposed a general conjecture claiming that the wake-up constant is reached with $n = 4$, in any norm. For instance, the conjectured wake-up constant for $\ell_p$ is $1 + 2^{1 + \max\set{1/p,1-1/p}}$, which holds for $p\in\set{1,2,\infty}$.
It turns out that $\OPT(n)$ is much smaller whenever $n$ is large, as they showed that $\OPT(n) = 3 + O(1/\sqrt{n}\,)$.


Polynomial-time approximation algorithms have been proposed, leading to the following upper bound on the wake-up constant $\OPTftp$: 57~\cite{ABFMS06}, 10.07~\cite{YBMK15}, 7.07~\cite{BCGH24}, 4.63~\cite{BGHO24}, 4.31~\cite{AAB25}. In the case of the plane equipped with the $\ell_1$ metric, a linear time algorithm has been proposed with a makespan bounded by~5 times the radius, which is best possible~\cite{BCGH24}.

FTP has also been studied in online and distributed contexts. In the online context, sleeping robots appear as time progresses~\cite{BW20,HNP06}. In the distributed context, robots do not have knowledge of the other robots' positions. They need to explore regions with a bounded vision and also limited communication capabilities (two awake robots can only communicate if they are within a certain distance)~\cite{GHLBM25}.

\begin{table}[htpb!]
    \centering
    \caption{Lower and upper bounds for the minimum makespan in the Euclidean plane.}
    \renewcommand{\arraystretch}{1.1}
    \begin{tabular}{ll||l|ll}
         &                                & \multicolumn{1}{c|}{$\OPT(n)$} & \multicolumn{1}{c}{$\OPTidn(n)$ \qquad (\textbf{our results})}          & \\ \hline\hline
         & Lower Bound $n=4$              & $1+2\sqrt{2} \approx 3.823$    & $2+2\sqrt{2} \approx 4.823$ \hfill [\cref{th:lowerbound-nsmall}] & \\
         & Lower Bound for any $n \geq 5$ & $3$                            & $4.672$* \hfill [\cref{th:lowerbound-nsmall}]                    & \\
        \cline{2-4}
         & UB (Convex case) for any $n$   & $1+2 \sqrt{2}$                 & $2+2 \sqrt{2} \approx 4.823$ \hfill [\cref{th:convex_UB}]          \\ &Upper Bound $n$ large & $3+O(1/\sqrt{n}\,)$ & $4.89+O(1/\sqrt{n}\,)$ \hfill [\cref{th:upperbound-nlarge}]  &\\
         & Upper Bound for any $n$        & $4.31$                         & $6.269=4.31+1.959$ \hfill [\cref{th:generalUB}]                  & \\

        \hline\hline
    \end{tabular}%
    \label{tab:bounds}\\
    {\footnotesize *We get a different bound  for \FTRPa: $\forall n\ge 5$, $\OPTan(n) \geq 1 +\frac{1}{4}(7^{1/4} + 7^{3/4})\sqrt{6} \approx 4.631$.}
\end{table}

\subparagraph{Contributions.}
Our goal is to compute and minimize the makespan for the \FTRP.
As for FTP, $\OPTid = \sup_{n \in \mathbb{N}} \OPTidn(n)$ is the \emph{wake-up and return constant} for the $\FTRPid$ and we use notation $\OPTa$ similarly for \FTRPa.
$\OPTidn(n)$ is the worst-case makespan for $n$ sleeping robots within a unit disk centered at the source in the \FTRPid model.


Our contributions can be summarized as follows:
\begin{itemize}
    \item[(1)]
        In \cref{sec:boundsL2}, we present several lower and upper bounds on the makespan of \FTRPid and \FTRPa in the Euclidean plane. In particular, we show that the overhead due to returns never exceeds $1.959 < 2=D$. We also remark that some distributions prove that the overhead must be at least $\sqrt{3} > 1.732$.
    \item[(2)]
        In \cref{sec:NPHard}, we show that computing the optimal makespan is NP-hard in the (unweighted) graph model setting for both variants \FTRPid and \FTRPa.
    \item[(3)]
        In \cref{sec:optimal}, we present three $2^{O(n)}$-time algorithms to compute the optimal makespan for the \FTP and both variants of the \FTRP for any non-negative distance function, therefore including Euclidean and metric spaces and weighted graph instances.
\end{itemize}

More precisely, in the Euclidean plane, \cref{th:lowerbound-nsmall} states that $\OPTid$ and $\OPTa$ are both at least~$2 + 2 \sqrt{2}$ and that this lower bound is achieved for $n = 4$ (see \cref{cor:upperboundSmall}). We show in \cref{th:convex_UB} that this is also an upper bound for any distribution of sleeping robots in convex positions.
Even if there are some pairs of robots at a pairwise distance $D = 2$, surprisingly, we show in \cref{th:generalUB} that the optimal makespan of any instance of \FTP and \FTRPid have a difference that never exceeds $1.959$. 

\cref{tab:bounds} summarizes the different known bounds for $\OPT(n)$, and the new ones we get for $\OPTidn(n)$. Note that each upper bound for $\FTRPid$ is also an upper bound for $\FTRPa$ and that each lower bound for $\FTRPa$ is also a lower bound for $\FTRPid$. Although we exhibit some specific distributions having different optimal makespan for $\FTRPa$ and $\FTRPid$, we are unable to prove a gap between the wake-up return constants $\OPTa$ and $\OPTid$, which deal with worst-case distributions.

In \cref{th:NPhard}, we show that even if we restrict the metric space to unweighted graphs of bounded degree, $\FTRPid$ and $\FTRPa$ are both NP-hard. Moreover, unless the ETH fails, no algorithm can solve any of these problems in time~$2^{o(n)}$.

Complementarily, in \cref{sec:optimal}, we provide algorithms returning the optimal makespan for \FTP, \FTRPid and \FTRPa in time $O(3^n n^2)$, $O(3^n n^3)$ and $O(9^n n^{3/2})$ respectively. These algorithms have been implemented, and used to verify our lower bounds for small $n$.


\section{Preliminaries}\label{sec:model}

In this section we give formal definitions of the key concepts introduced above.
An instance is defined by a multi-set of points $\P = \range{\pos_0}{\pos_n}$ within a metric space $(X, d)$, each point $\pos_i$ being associated with a robot $r_i$ initially located at $\pos_i$. Depending on the context, $\pos_r$ also stands for the initial position of robot $r$.  
We assume that $r_0$ is the initially awake robot. 

The solution of an $\FTP$ instance can be encoded by a \emph{wake-up tree} $\wutree=(\P, \wakeuparc)$, a weighted tree rooted at $\pos_0$ that spans $\P$.
The existence of an edge $(\pos_i, \pos_j) \in \wakeuparc$ means that robot $r_j$ is woken up by an awake robot previously located at node $\pos_i$.
In other words, $r_j$ is woken up by $r_i$ or by the robot~$r'$ that moved to $\pos_i$ and woke up $r_i$. 
As a result, the root has at most one child, and any node has up to two children.
Arcs of $\wakeuparc$ are called \emph{wake-up arcs} and are oriented from parent to child.
Arcs' weights are given by the metric distance $d$ between the positions of their endpoints, which can be zero since we allow co-located robots.

\begin{figure}[htbp!]
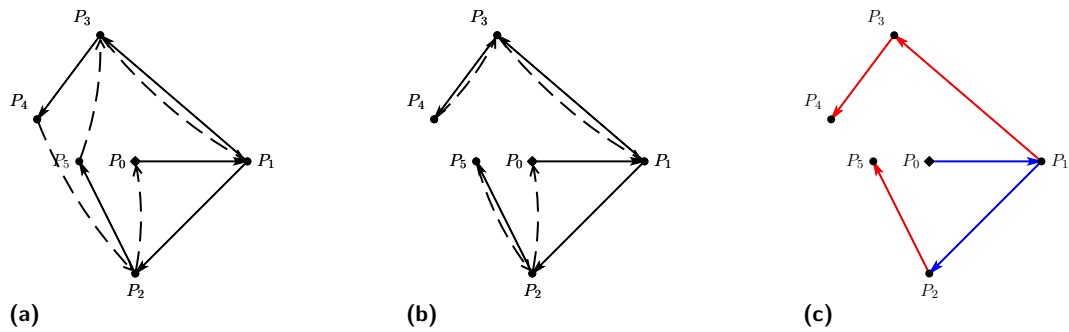

    \centering
    \def\SPC{\vspace{-4ex}}
    \begin{subfigure}[b]{0.25\linewidth}
        \incfig{\FigExCov}{1}
        \SPC\caption{}
        \label{fig:ano-tree}
    \end{subfigure}
    \hfill
    \begin{subfigure}[b]{0.25\linewidth}
        \incfig{\FigExId}{1}
        \SPC\caption{}
        \label{fig:id-tree}
    \end{subfigure}
    \hfill
    \begin{subfigure}[b]{0.25\linewidth}
        \incfig{\FigExIdBi}{1}
        \SPC\caption{}
        \label{fig:bicolor-tree}
    \end{subfigure}
    \caption{
        (a) A \FTRPa wake-up strategy (self-loops on $\pos_4$ and $\pos_5$ are omitted).
        (b) A \FTRPid strategy:  $r_0$ wakes up $\{r_1,r_2\}$ and comes back, $r_1$ wakes up $\{r_3,r_4\}$ and comes back, $r_2$ wakes up $r_5$ and comes back. The robots $r_3, r_4$ and $r_5$ do not move.
        (c) The same \FTRPid wake-up strategy is represented with a bicolored tree: the return arcs are between endpoints of maximal monochromatic sub-paths.}%
    \label{fig:example}
\end{figure}

Given a wake-up tree $\wutree$, the \emph{wake-up time} of a robot $r_i$ (resp. position $\pos_i$) is denoted $\wutime_\wutree(r_i)$ (resp. $\wutime_\wutree(\pos_i)$) and corresponds to the weighted length of the path from $\pos_0$ to $\pos_i$ in~$\wutree$.
Thus the makespan of a wake-up strategy encoded by $\wutree$ is the weighted depth of $\wutree$.
%
%
%

A solution to an instance $\P$ of $\FTRPa$ can be encoded by a tree being an extension of the wake-up tree duplicating the initial positions and adding the clones as leaves.
For convenience, we represent a solution with a pair $(\wutree, \returnarc)$ where $\wutree$ is a wake-up tree of $\P$ and $\returnarc \subseteq \P^2$ is a set of \emph{return arcs}, such that each $\pos_i \in \P$ is the destination of exactly one return arc and the graph $(\P, \wakeuparc\cup \returnarc)$ is an oriented Eulerian graph (i.e. each vertex has as many outgoing arcs as incoming arcs).
Self-loops are allowed and correspond to a robot staying still during the process (see \cref{fig:ano-tree}).
The \emph{\completionTime} of $r_j$, is denoted $\returntime(r_i)$,  is $\wutime_\wutree(r_j) + d(\pos_j, \pos_i) $ where $(\pos_j, \pos_i) \in \returnarc$.
The \emph{makespan} of a pair $(\wutree, \returnarc)$ is the maximum over all \completionTime{}s.

A solution $(\wutree, \returnarc)$ for $\FTRPa$ is also a solution for $\FTRPid$ if each robot $r_i$ returns to its initial position $\pos_i$ after waking up the other robots.
It is the case if each edge of the graph $(\P, \wakeuparc \cup \returnarc)$ belongs to a unique cycle.
The solution depicted on~\cref{fig:id-tree} satisfies this property, while the one on~\cref{fig:ano-tree} does not.

Unlike for \FTP or \FTRPa where robots are indistinguishable, defining a strategy in \id{} context requires specifying each robot's \emph{trajectory}.
For $r_i \in \R$ we define $\last_{r_i} \in \P$ as the last position visited by $r_i$ before returning to $\pos_i$.
As an example, given a $\FTRPid$ solution encoded by a pair $(\wutree, \returnarc)$, $r_i$'s trajectory is the unique cycle of $(\P, \wakeuparc \cup \returnarc)$ containing the return arc that ends at $\pos_i$.
This arc is $(\last_{r_i}, \pos_i)$ and is called \emph{$r_i$'s return arc}.

Conversely, given a trajectory for each robot (seen as a sequence of points of $\P$), it is immediate to get the corresponding wake-up tree, where each trajectory (minus the return arc) will still appear as a branch of that tree.
There is a one-to-one correspondence between robots and trajectories by introducing a bicoloring of the tree arcs (as depicted on \cref{fig:bicolor-tree}). 
Most strategies depicted in the present article are represented using bicolored trees.
The return arc of a robot $r_i$ always goes from the leaf of the monochromatic path rooted in $\pos_i$, which is $\last_{r_i}$, to $\pos_i$.

\section{Lower and Upper Bounds for Euclidean Plane}%
\label{sec:boundsL2}
\subsection{Lower bounds}

As soon as $n \ge 3$, getting  a lower bound on $\OPTidn(n)$ and $\OPTan(n)$ is not straightforward. We only give here some hints to explain how we get our bounds. Informally, the construction of the lower bounds is based on regular convex configurations for $n=3$ and $n = 4$ and on properties of the wake-up trees (number of leaves and the minimum \completionTime{} for the robots reaching the leaves).

\begin{figure}[htbp!]
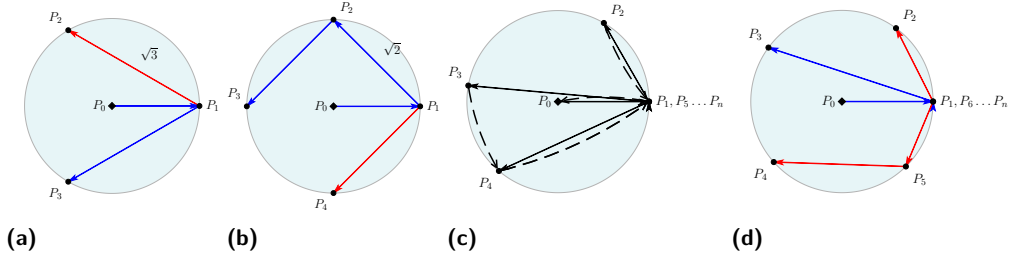

    \centering
    \def\SPC{\vspace{-3ex}}
    \begin{subfigure}[b]{0.2\textwidth}
        \incfig{\FigRegularOptBiIdThreeA}{1}
        \SPC\caption{}%
        \label{fig:regular_convex_3}
    \end{subfigure}
    \begin{subfigure}[b]{0.2\textwidth}
        \incfig{\FigRegularOptBiIdFourA}{1}
        \SPC\caption{}%
        \label{fig:regular_convex_4}
    \end{subfigure}
    \begin{subfigure}[b]{0.26\textwidth}
        \incfig{\FigWorstCovA}{1}
        \SPC\caption{}%
        \label{fig:worstcaseCov}
    \end{subfigure}
    \begin{subfigure}[b]{0.26\textwidth}
        \incfig{\FigWorstIdD}{1}
        \SPC\caption{}%
        \label{fig:worstcaseId}
    \end{subfigure}
    \caption{(a) and (b): Worst distributions for $n = 3$ and $n = 4$, and their optimal wake-up trees for \FTRPid. (c): Asymptotic lower bound for \FTRPa. (d): Asymptotic lower bound for \FTRPid. In these two last configurations, there are respectively $n-3$, $n-4$ robots at $\pos_1$.}%
    \label{fig:lbmalln}
\end{figure}

\begin{restatable}[Lower bounds]{theorem}{lowerboundnsmall}%
    \label{th:lowerbound-nsmall}
    The following inequalities hold true.\vspace{-2ex}
    \begin{itemize}[noitemsep]
        \item $\OPTidn(3)\geq \OPTan(3) \geq 1+2\sqrt{3} \approx 4.46,$  \item $\OPTidn(4)\geq \OPTan(4) \geq  2+2\sqrt{2} \approx 4.824,$
        \item For any $n \geq 5$, $\OPTan(n) \geq 1 + \frac{1}{4} (7^{1/4} + 7^{3/4})\sqrt{6} > 4.631$ and $\OPTidn(n) > 4.672.$
    \end{itemize}
\end{restatable}

It turns out that these lower bounds are tight for $n=3$ and $n=4$ (as we show in the next section).
For $n \geq 5$, a new phenomenon appears: the lower bounds (and the underlying distributions of robots) are not the same for $\OPTidn(n)$ and $\OPTan(n)$. Moreover, the worst distribution of robots is no longer regular. For instance, to lower bound $\OPTan(n)$ (see \cref{fig:worstcaseCov}) whenever $n \geq 5$, we spread the $n$ sleeping robots on $4$ sites: $n-3$ sleeping robots are located on the unit circle with the same angle $\alpha_1 = 0$ ($\pos_1=\pos_5=\ldots=\pos_{n}$) whereas the angles of the $3$ remaining robots are $\alpha_2 = 2\arctan\pare{\frac{7^{1/4} \sqrt{6}}{3\sqrt{7}-1}}$ and
$\alpha_3 = \alpha_2+ 2\arctan\pare{\frac{7 ^{1/4} \sqrt{6}}{(\sqrt{7} - 1)(4 \sqrt{7} - 6)}}$ and $\alpha_4 =\alpha_3 + \alpha_2$.
In this configuration, regardless of the value of $n$, at most 2 robots positioned at $P_1$ need to move to wake the other robots. It is therefore sufficient to analyze the case $n=5$, which is done by evaluating a finite number of wake-up trees.
See Appendix~\ref{appendix:lowerbound-nsmall} for the full proof of \cref{th:lowerbound-nsmall}.


\subsection{Upper bounds}

\subsubsection{The convex case}\label{sec:convexCase}

In this section we focus on the case where sleeping robots are in a convex configuration within a unit disk.
We provide a dedicated algorithm, \convex and we show that it achieves a makespan of $2 + 2\sqrt{2}$ for the \FTRPid.

\begin{restatable}{theorem}{convexUB}\label{th:convex_UB}
    For any unit-disk instance\footnote{That is a distribution of robots in the unit disk of $(\mathbb{R}^2,\ell_2)$ where the initial awake robot is at the origin.} $\P$ where the sleeping robots are in convex position, the makespan of \FTRPid for $\P$ is at most $2 + 2\sqrt{2}$.
\end{restatable}

Assuming that positions can be ordered cyclically with respect to their angular distances to a fixed point, we sketch Algorithm \convex (see \cref{fig:convex-algo}) used for proof:
(1) First, we show that the worst distribution is obtained whenever the sleeping robots lie on the boundary of the unit circle.
(2) We then define a simple algorithm \linear to wake up \emph{consecutive} robots according to a cyclic ordering: informally, robot $r_i$ located at position $\pos_i$ wakes up $r_{i+1}$ and returns to $\pos_i$.
(3) Given a specific geometric position $Q$, we partition the circle into $3$ arcs $(\pos_1,Q), (Q,\pos_{k})$ and $(\pos_{k+1},P_1)$ with respect to the maximum empty cone of angle $\beta$ between points $\pos_k$ and $\pos_{k+1}$.
Point $\pos_1$ is diametrically opposed to the widest empty cone of measure $\beta$.
Robot $r_0$ wakes up the extreme robots $r_1,r_2$ and $r_j$ and returns to $\pos_0$.
(4) In each arc, algorithm \linear is applied.
The analysis is based on a precise analysis of this algorithm and the position of $Q$.

As a corollary of \cref{th:convex_UB}, we can get upper bounds for $\OPTidn(n)$ and $\OPTan(n)$ for $n \in \{3,4\}$ by adding a specific algorithm for the non-convex distributions for $n = 4$ (see the proof in Appendix~\ref{appendix:UBConvex}). According to \cref{th:lowerbound-nsmall}, these upper bounds are best possible.


\begin{restatable}{corollary}{upperboundSmall}\label{cor:upperboundSmall}
    It holds that $\OPTan(3) = \OPTidn(3) = 1 + 2\sqrt{3}$ and $\OPTan(4) = \OPTidn(4) = 2 + 2\sqrt{2}$.
\end{restatable}

\subsubsection{Asymptotic upper bound}

We adapt the \FTP algorithm from~\cite[Theorem~2]{BCGH24} (cone+return+cone). In this version, the disk is divided into $\sqrt{n}$ cones of equal angle. First, we wake up the robots in the densest cone. This cone contains at least $\sqrt{n}$ robots. For each remaining cone, we use one awake robot. This algorithm is called \splitconestrategy. It runs in $O(n \log n)$~time and produces a makespan of at most $3 + O(1/\sqrt{n}\,)$ for the FTP problem. Unfortunately, for the FTRP problem, the makespan of this algorithm can reach $5+ O(1/\sqrt{n}\,)$ (for example, for $n$ points regularly distributed on the unit circle). Indeed, the robot that must wake up the cone opposite to its initial position might have a return arc of length up to $2$.

Our new algorithm relies on two ideas. The first idea is to use the initial robot to wake up the cone $C_b$ opposite to the initial cone and assign it a larger opening $\alpha_b\approx 0.9439$. The second idea is to use a second robot if $C_b$ contains another robot at a distance of at most $\rho_c \approx 0.9445$ from $\pos_0$. With these two ingredients, we can reduce the makespan to $4.89 + O(1/\sqrt{n}\,)$.

\begin{restatable}{theorem}{upperboundnlarge}\label{th:upperbound-nlarge}
    For any unit-disk instance with $n$ sleeping robots, we can build in time~$O(n\log{n})$ an \FTRPid solution with a makespan of $4.89 + O(1/\sqrt{n}\,)$.
\end{restatable}

\subsubsection{General upper bound}%
\label{sec:generalUB}

In this section, we present the general ideas of the proof of Theorem~\ref{th:generalUB}.
The detailed proofs are deferred to Appendix~\ref{appendix:GUB}.

\begin{restatable}{theorem}{generalUB}%
    \label{th:generalUB}
    Let $t$ be the makespan of a solution of the \FTP for a unit-disk instance $\P$. There exists a solution of makespan $t_R$ for the \FTRPid on $\P$ such that $t_R < t + 1.9583$.
    In particular, for every $n\in\mathbb{N}$, $\OPTidn(n) < \OPT(n) + 1.9583$.
\end{restatable}

We observe that the constant $1.9583$ in~\cref{th:generalUB} cannot be replaced by any constant less or equal to $\sqrt{3} \approx 1.732$. Indeed, consider the worst-case distribution for $\OPTidn(3)$, where the~$3$ robots are pairwise at distance $\sqrt{3}$ (cf. \cref{fig:regular_convex_3}). Clearly, such distribution has a solution for \FTP with makespan $t \le 1 + \sqrt{3}$, whereas the makespan of any solution for \FTRPid or \FTRPa must be at least $t_R \ge 1 + 2\sqrt{3} \ge t + \sqrt{3}$ (cf. \cref{th:lowerbound-nsmall}).

Let us sketch the proof of \cref{th:generalUB}.
We show that there is some $\eps > 0$ such that, for each distribution of robots $\P$, for which there is a wake-up tree $\wutree$ with makespan $t$ solving \FTP on $\P$, there is also a wake-up tree and return arcs for \FTRPid with makespan at most $t + 2 - \eps$.
The concrete $\eps$ for which we will show the statement will be $0.0417$, but most of our arguments apply for possibly way larger values of $\eps$.

Starting with a wake-up tree $\wutree$ for $\P$, we can consider any solution $(\wutree, \returnarc)$ for $\FTRPid$.
For each robot $r \in \R$, its \completionTime{} is $\wutime_\wutree(\last_{r}) + |\last_{r} \pos_r|$.
Let~$(r,r')$ be a pair of robots such that $\pos_{r'}$ is an ancestor of~$\pos_r$ in $\wutree$.
We say that~$(\pos_r,\pos_{r'})$ is an~\emph{$\eps$-good pair}, if $\wutime_\wutree(\pos_r) + |\pos_r\pos_{r'}| \leq t + 2 - \eps$.
Otherwise, we call~$(\pos_r,\pos_{r'})$ an~\emph{$\eps$-bad pair}.
Similarly, we call a return arc~$(\last_r,\pos_r)$ \emph{$\eps$-good} if~$(\last_r,\pos_r)$ is an $\eps$-good pair, and \emph{$\eps$-bad} otherwise.
We denote by $\Bad$ the set of all robots $r$ such that $(\last_r,\pos_r)$ is an $\eps$-bad return arc.

The overall goal of the proof is to build a wake-up tree $\wutree'$ for \FTP, not identical to $\wutree$, and a set of return arcs $\returnarc'$ such that there are no $\eps$-bad return arcs in $(\wutree',\returnarc')$.
It will be necessary to alter $\wutree$ in order to remove $\eps$-bad return arcs.
For the construction, it is relevant to start with a solution $(\wutree, \returnarc)$ that minimizes the number of $\eps$-bad return arcs among all solutions for $\FTRPid$ based on $\wutree$.
Given a set of return arcs $\returnarc$, we can define the $\eps$-cost of $\returnarc$ as the sequence $(b_k)_{k\in\mathbb N}$ where $b_k$ is the number of robots $r \in \Bad$ such that the path from $\pos_{r_0}$ to $\pos_r$ in $\wutree$ has $k$ edges. 
In the following, we consider that $\returnarc$ minimizes the number of $\eps$-bad return arcs and, among all such solutions, has the lexicographic largest $\eps$-cost.

\subparagraph{Construction of $\wutree'$.}
The first part of the construction is disconnecting branches of $\wutree$ as low as possible, so that there are no $\eps$-bad return arcs in the remaining wake-up tree.
The resulting wake-up tree does not completely span $\P$, meaning we will have to complete the tree later.
Informally, for every $r\in \Bad$, we let $r$ follow the wake-up tree as long as its return does not induce an $\eps$-bad return arc, and then come back to its initial position.
For every $r\in \Bad$, let~$\good_r$ be the closest ancestor of~$\last_r$ in~$\wutree$ for which~$(\good_r,\pos_r)$ is an $\eps$-good pair.
Such an ancestor exists since~$(\pos_r,\pos_r)$ is guaranteed to be a~$\eps$-good pair.
Moreover, let~$\bad_r$ be the child of~$\good_r$ leading to~$\last_{r}$, and let $\wutree_{\bad_r}$ be the subtree of~$\wutree$ rooted in~$\bad_r$.
We denote by~$\Lset_r$ the set of the robots with initial positions in $\wutree_{\bad_r}$. 
The first operation is disconnecting~$\Lset_r$ from $\wutree$ and adding the $\eps$-good return arc $(\good_r,\pos_r)$ to $\returnarc'$.

To state our results about the robots in~$\Lset_r$, we define a \emph{crown of angle $\alpha$ and width $w$} as the region containing all points between distance $1-w$ and $1$ from $\pos_0$ and whose angular coordinates are within a cone of angle $\alpha$ centered at $\pos_0$.
Moreover, we call a crown of width $\eps$ and of angle $2\pi$  an~\emph{$\eps$-crown}.
\begin{restatable}{lemma}{lemmaOnCrown}%
    \label{lem:rb on crown}
    Let~$(\last_r,\pos_r)$ be an $\eps$-bad return arc in~$\wutree$ and let~$x \in \Lset_r$.
    Then, $\pos_x$ is on the $\eps$-crown and has distance at most~$\eps$ to some leaf of the subtree of~$\wutree$ rooted in~$\bad_r$.
\end{restatable}

After the previously mentioned trimming operation, robots of $\Lset=\bigcup_{r\in \Bad} \Lset_r$ are disconnected from the wake-up tree. 
The second step is to find a strategy for the awakening of these robots.
\cref{lem:rb on crown} guarantees that for each $r$, the robots of~$\Lset_r$ are located within an $\eps$-crown.
For every $r\in \Bad$, the idea is to designate one robot $r' \in \R\setminus \Lset$ that was not disconnected
during the first operation, and such that $\pos_{r'}$ is within the $\eps$-crown.
This robot, after being awakened, wakes up the robots of $\Lset_r$ using a dedicated $\FTP$ algorithm.
Since the returns are confined to a small sub-crown of width $\eps$, this step will not significantly increase the makespan as stated in the following lemma.

\begin{restatable}{lemma}{lemmaCrown}%
    \label{cone wustrategy}
    Let $C$ be a crown of angle $\alpha$ and width $w$.
    Assuming that the initial awake robot is anywhere inside $C$, there exists a strategy for solving $\FTRPid$ in $C$ with a makespan $M(C) \leq 2 \alpha + w \cdot (4+\varphi)$ where $\varphi = (1+\sqrt{5})/2$ is the Golden ratio.
\end{restatable}

Given $r\in\Bad$, we first define $\Rbest$ and its initial position $\best_r$ under the assumption that~$\eps < 1$.
We call $\best_r$ the \emph{best ancestor of $r$} if $\best_r$ is the closest ancestor of $\pos_r$ in $\wutree$ such that $(\last_{\Rbest},\best_r)$ is an $\eps$-good return arc and such that the robot initially located at $\best_r$ wakes up $r$ or a robot located at an ancestor of $\pos_r$.
We now define $r'$ as being the robot initially located at $\last_\Rbest$.

The makespan of robots in $\mathcal R \setminus \Lset$ is at most $t + 2 - \varepsilon$ since by construction there are no $\varepsilon$-bad return arcs.
On the other hand, Lemmas~\ref{cone wustrategy} and~\ref{small cone contains Lr bis} guarantee that the makespan of robots in $\Lset$ is at most~$t + 2\alpha_\eps + (4 + \varphi) \eps$.
By setting~$\eps = 0.0417$, we obtain a makespan dominated by~$t + 1.9583$ in both situations, which completes the proof of \cref{th:generalUB}.

\begin{figure}[htbp!]
    \centering
    \begin{subfigure}[b]{0.35\textwidth}
        \includegraphics[width=\textwidth]{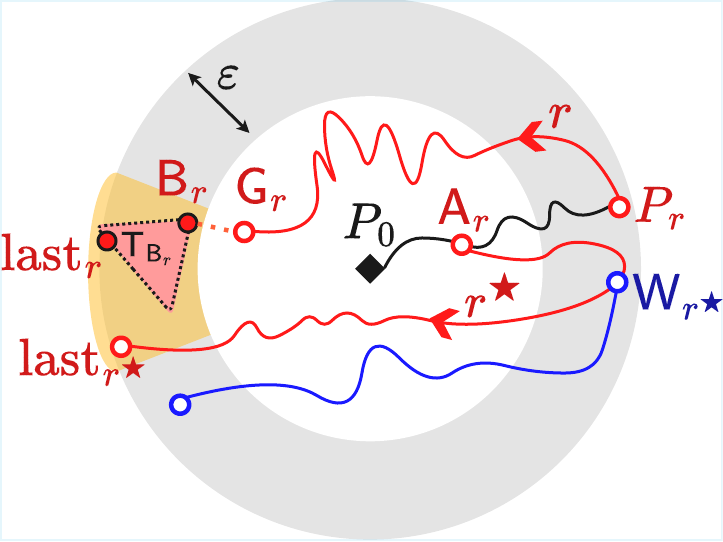}
        \caption{}%
        \label{fig:GUB-2}
    \end{subfigure}
    \hfill
    \begin{subfigure}[b]{0.35\textwidth}
        \includegraphics[width=\textwidth]{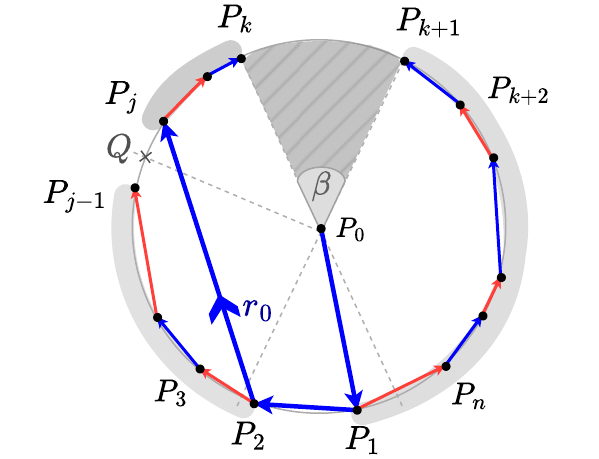}
        \caption{}%
        \label{fig:convex-algo}
    \end{subfigure}
    \caption{(a) Robots $r$, $\Rbest$ and $\Rworst$ are initially located at positions $\pos_r, \best_r$ and $\worst_{\Rbest}$.
        Robots of $\Lset_r$ within the highlighted subcrown of angle $\alpha$ are awakened from a robot located at $\last_{\Rbest}$.
        $\last_{\Rbest}$ and leaves of $\wutree_{\bad_r}$ are at distance at least $2-\epsilon$ from node $\worst_{\Rbest}$.
        (b) Bi-colored wake-up tree resulting from the Convex Algorithm.
        Highlighted arcs are woken up using \linear}
\end{figure}

\begin{restatable}{lemma}{lemSmallCone}%
    \label{small cone contains Lr bis}
    Let~$\best_r$ be the best ancestor of at least one robot $r$. 
    There is a crown with angle~$\alpha_\eps=2 \arcsin{\left(\frac{\eps}{2(1-\eps)}\right)} + 4 \arccos(1-\frac{\eps}{2})$ and width $\varepsilon$ that contains~$\last_{\Rbest}$ and the initial position of each robot in~$\bigcup_{r\in \Bad_\Rbest} \Lset_r$.
\end{restatable}

\section{NP-Hardness for Unweighted Graphs}%
\label{sec:NPHard}





\newcommand{\var}{\mathbf{var}}
\newcommand{\SAT}{\textsc{3-SAT}\xspace}
\newcommand{\LSAT}{\textsc{SAT}\xspace}
\newcommand{\Oh}{\mathcal{O}}

\newcommand{\prob}[3]{\begin{quote}  \textsc{#1}\\  \textbf{Input:} #2\\  \textbf{Question:} #3\end{quote}}

In this section, we study the two variants of \FTRP on graphs.
The space is defined by a graph $G$ with non-negative weights on the edges. Robots are located on vertices of the graph and move along edges.
A robot sleeping on a vertex $v$ is woken up as soon as an awake robot reaches~$v$.


\begin{restatable}{theorem}{NPhard}\label{th:NPhard}
    \FTRPid and \FTRPa are NP-hard on unweighted graphs of constant maximum degree and logarithmic diameter. 
    Moreover, unless the ETH fails, neither of these problems on graphs can be solved in $2^{o(n+m)}$~time, where~$n$ and~$m$ denote the number of vertices and edges, respectively.
\end{restatable}

We reduce from a variant of \SAT where we assume that each literal occurs at most twice.
This variant is NP-complete and cannot be solved in $2^{o(|F|)}$~time, unless the ETH fails, as the reduction by Tovey~\cite{Tovey84} uses a linear space reduction. 
This property allows us to describe a construction of bounded degree and where the number of vertices is linear in the size of the formula $F$.

Let~$F$ be an instance of this \SAT variant with variable set~$X$.
For technical reasons, we add for each variable~$x\in X$ the clause~$(x\lor \overline{x})$ to the formula.
We denote by $F'$ the resulting formula, and by $C$ the clauses of $F'$.
It is clear that~$F$ is satisfiable if and only if~$F'$ is satisfiable.
Note that each variable occurs at most three times positively and at most three times negatively in~$F'$.

\subparagraph{Construction.}
We build a graph $G = (V,E)$, an instance of the FTRP on an unweighted graph.
Let~$n = |X|$.
For the sake of simplicity, we assume that~$n$ is a power of~$2$.

\begin{figure}[htbp!]
    \centering
    \begin{subfigure}[b]{0.4\textwidth}
        \centering
        \includegraphics[width=\textwidth]{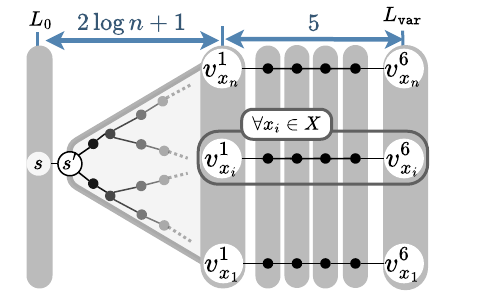}
        \caption{Set-up gadget.}%
        \label{fig:reduction-setup}
    \end{subfigure}
    \hfill
    \begin{subfigure}[b]{0.4\textwidth}
        \centering
        \includegraphics[width=\textwidth]{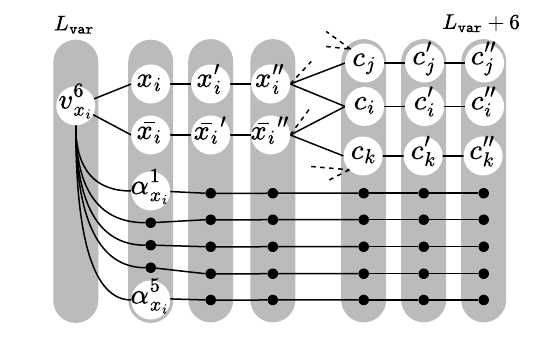}
        \caption{Variables and clauses gadget.}%
        \label{fig:reduction-vars}
    \end{subfigure}
    \caption{An example for the gadget constructions.}%
    \label{fig:NP-construction}
\end{figure}

\emph{Set-up gadget.} We start with a vertex~$s$ at which the source robot is initially located, and the graph obtained from subdividing each edge of a balanced binary tree rooted at a vertex $s'$ with $n$ leaves.
The leaves are denoted by $a_x$ for each $x \in X$.\footnote{Note that this implies that each leave~$a_x$ has distance exactly~$2\cdot \log n + 1$ from~$s$.}
We connect $s$ to $s'$.
Next, for each $x \in X$ we add a path~$(a_x =:  v^1_x,v^2_x,v^3_x,v^4_x,v^5_x,v^6_x =: v_x)$.
We get a binary tree rooted at~$s$ having $n$ leaves that has height $2\log n + 6$.
See \cref{fig:reduction-setup}.

\emph{Variable gadget.} For each variable~$x \in X$, we add two paths $(x, x', x'')$ and $(\overline{x}, \overline{x}', \overline{x}'')$, and we connect~$x$ and~$\overline{x}$ to~$v_x$.
These two paths will encode the truth assignment of $x$ depending on which of the robots on~$x$ or~$\overline{x}$ is woken up first.
For each~$i \in [1, 5]$, connect to~$v_x$ a path on $6$ vertices whose first vertex (the one adjacent to~$v_x$) is denoted by~$\alpha^i_x$ and whose last vertex is denoted by~$\omega^i_x$.
The distance from $s$ to any~$\omega^i_x$ is exactly $2\log n + 12$.

\emph{Clause gadget.} For each clause $c \in C$, add a path~$(c, c',c'')$.
For each literal $x$ ($\overline{x}$) occurring in clause $c$, there's an edge connecting $c$ and $x''$ ($\overline{x}''$).
Since each clause contains at least one literal, for $c \in C$, the distance from $s$ to $c''$ is $2\log n + 12$.
See \cref{fig:reduction-vars}.

\emph{Notations.}
We define the \emph{depth} of $G$, $d := 2 \log n + 12$, which is the eccentricity of $s$.
Let~$0 \leq i \neq j \leq d$.
We denote by $L_i$ the \emph{$i$-th layer} of the graph, formally $L_i := \{v \in V \mid |sv| = i\}$ where $|sv|$ denotes the length of a shortest path from $s$ to $v$.
We say that a vertex~$u\in L_i$ is an \emph{ancestor} of~$v\in L_j$, if the distance from~$u$ to~$v$ is~$j-i$.
Note that this is the case if there is a shortest path from~$s$ to~$v$ that contains~$u$.
Furthermore, note that~$L_\var = \{v_x\mid x\in X\}$, where~$\var := d-6 = 2\log n + 6$.
Moreover, note that the number of vertices of the graph is linear in the size of the formula~$F'$ and thus linear in the size of the formula~$F$.
\subparagraph{Proof summary.} The correctness of the reduction is shown by proving the following.

\begin{restatable}{lemma}{NPcorrectness}\label{lem:NP-correctness}
    If~$F'$ is satisfiable, then there is a $\FTRPid$ solution of makespan~$d+1$.
    Moreover, if there is a $\FTRPa$ solution of makespan~$d+1$, then~$F'$ is satisfiable.
\end{restatable}

See \cref{app:NP-hardness} of the Appendix for full proof.
As each wake-up tree for $\FTRPid$ is also a solution for $\FTRPa$, this then shows that both problems are NP-hard.
Moreover, no wake-up tree with any kind of return can have a makespan smaller than~$d+1$.

Our arguments rely on the fact that robots in layer $L_d$ have to be woken up in a time that is exactly $d$, which is their distance from the root vertex.
This strongly constrains the robot's trajectories, as almost every robot must follow a shortest path.
The graph structure also ensures that only a fixed number of robots can reach some critical layers, as the number of positions to return to within a fixed radius is bounded.

First, the set-up gadget ensures that robots of the $n$ vertices~$a_x$ are all woken up at time~$2\log n + 1$.
Each of these robots follows their respective shortest path to the variable gadget and arrives at time~$\var$ at vertex~$v_x$, with $6$ other robots woken up along the path $a_xv_x$.
We use a counting argument and the distance to vertices in layer~$L_d$ to show that five of these robots are dedicated to the five $\alpha_x^i\omega_x^i$ paths (with~$i\in [1,5]$).
The sixth robot wakes up the robot on one of the vertices~$x$ or~$\overline{x}$ at time~$\var + 1$, which models an assignment.
Conversely, we ensure that at least one of them is woken up at~$\var+1$, as~$(x\lor \overline{x})$ is a clause of~$F'$.
Based on the structure of the clause and variable gadgets, this ensures that for a truth assignment, at time~$\var+1$ the robot on at least one literal of~$c$ is woken up, as otherwise, no robot can reach~$c''$ at a time prior to~$d+1$.



\newcommand\proofHardness{
    \subparagraph{Road map.}
    The idea behind the reduction is as follows:
    The first~$2\log n + 1$ layers are designed in a way, such that all robots that are initially located on these layers can be woken up and return to their initial position at time at most~$2\log n + 2 < d+1$, while ensuring that robots of the vertices~$a_x$ with~$x\in X$ are all woken up at time~$2\log n + 1$.
    These robots then follow their respective shortest paths to the variable gadget, that is, to the vertex~$v_x$ for each~$x\in X$.
    The robots woken up on these paths follow toward the variable gadget.
    At time~$\var$, we will then have~$6$ awake robots located at each vertex~$v_x$.
    For each~$x\in X$, five of these robots wake up the robots on the vertices~$\alpha_x^i$ with~$i\in [1,5]$, whereas the sixth robot can only wake up the robot on one of the vertices~$x$ or~$\overline{x}$ at time~$\var + 1$.
    The set of literals for which the respective vertex is woken up at time~$\var + 1$ then models a truth assignment.
    Moreover, this truth assignment satisfies formula~$F'$ if and only if the overall makespan with return is at most~$d+1$.
    For the anonymous return (\FTRPa), this wake up strategy is in fact necessary:
    To obtain a makespan with return of~$d+1$, each robot initially located at a vertex of~$L_d$ needs to be woken up at time~$d$.
    This holds in particular for the robot on vertex~$c''$ for each clause~$c\in C$.
    Based on the structure of the clause and variable gadgets, this ensures that at time~$\var+1$ the robot on at least one literal of~$c$ is woken up, as otherwise, no robot can reach~$c''$ at a time prior to~$d+1$.
    The only remaining obstacle is then to ensure that for each variable only one robot on one of its literals is woken up at time~$\var+1$.
    Intuitively, this holds due to two facts:
    Firstly, for each variable~$x$, at least one of the robots located on a literal of~$x$ needs to be woken up at time~$\var+1$, as~$(x\lor \overline{x})$ is a clause of~$F'$.
    Secondly, no more than~$6n$ robots can wake up robots on vertices of layer~$L_{\var + 1}$ at time~$\var +1$, as otherwise, not all of these robots can return to pairwise distinct positions in the remaining time until the end of the makespan of~$d+1$.
    This is essentially due to the fact that the layers~$L_\ell$ with~$\ell\in [\var - 5, \var]$ together only contain~$6n$ vertices, and each vertex needs to contain exactly one robot at time~$d+1 = \var + 7$.
    By the fact that~$5n$ of these robots are necessary to wake up the robots initially located on the vertices of~$\{\alpha_x^i\mid x\in X, \in [1,5]\}$, for only one literal per variable can the initially located robot be woken up at time~$\var+1$.

    We now prove formally the correctness of the construction.
    \NPcorrectness*
    \begin{proof}
        $(\Rightarrow)$
        Assume that~$F'$ is satisfiable and fix an arbitrary satisfying assignment.

        \begin{figure}[htbp!]
            \centering
            \includegraphics[width=0.7\textwidth]{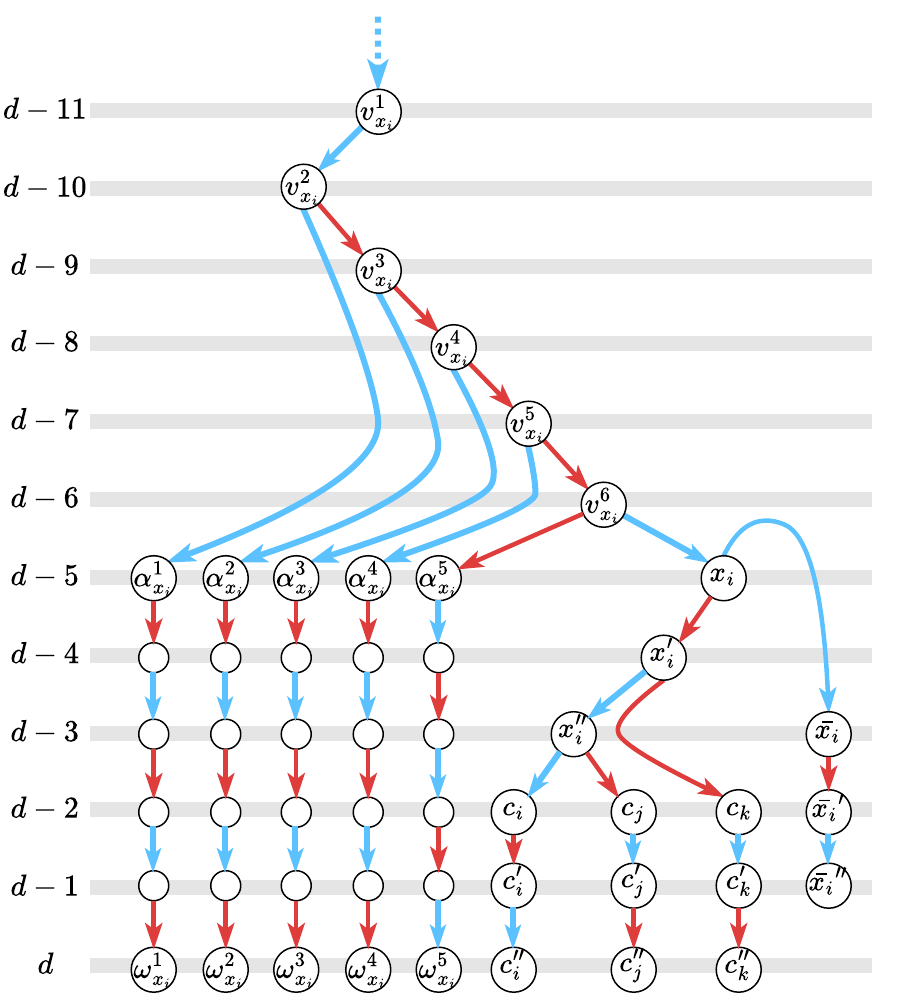}
            \caption{Bi-color wake-up tree for the variable and clause gadget.
                Curved arcs to hints the path taken in the graph}%
            \label{fig:NP_wakeup}
        \end{figure}

        Let us start by describing the strategy for waking up every robot at distance at most~$2\log n + 1$ from~$s$, within~$\log n + 1$ time steps.
        Initially, $s$ wakes up~$s'$.

        Let~$i\in [1,2\log n-1]$ be odd and let~$v\in L_i$ be a vertex.
        Then, by construction, $v$ has exactly two neighbors in layer~$L_{i+1}$.
        If the robot~$r$ located on vertex~$v$ is woken up by a robot~$r'$, then~$r$ and~$r'$ move to the two neighbors of~$v$ in~$L_{i+1}$, wake up the sleeping robots there, and return to their initial position.

        Let~$i\in [2,2\log n]$ be even and let~$v\in L_i$ be a vertex.
        Then, by construction, $v$ has exactly one neighbor in layer~$L_{i+1}$.
        If the robot~$r$ located on vertex~$v$ is woken up by a robot~$r'$, then~$r'$ immediately returns to its initial position and~$r$ moves and wakes up the sleeping robot on the unique neighbor of~$v$ in~$L_{i+1}$.
        If~$i = 2\log n$, then~$r$ returns to its initial position immediately.

        In this way, we ensure that each robot initially located on a vertex of the layers~$L_i$ with~$i\in [0,2\log n]$ returns to its initial position at the latest at time~$2\log n + 2 < d$, since each such robot traverses at most~$4$ edges each.
        Moreover, the strategy ensures that each robot initially located on a vertex of layer~$L_{2\log n + 1}$ is woken up at time~$2\log n + 1$.

        That is, for each variable~$x\in X$, the robot initially located at~$a_x$ is woken up at time~$2\log n + 1$.
        This robot then moves along the unique shortest path $(v_x^1,v_x^2,v_x^3,v_x^4,v_x^5,v_x^6 = v_x)$ to vertex~$v_x$ and wakes up all robots on the way.
        All these robots also follow the path until vertex~$v_x$.
        Hence, at time~$\var$, there are~$6$ awake robots at vertex~$v_x$.
        For each~$i\in [1,6]$, let~$r_x^i$ denote the robot initially located at~$v_x^i$.

        For each~$i\in [1,5]$, robot~$r_x^i$ now moves to vertex~$\alpha_x^i$, wakes the robot up on that vertex, and afterwards returns to its initial position.
        Hence, $r^x_i$ arrives at its initial position at time at most~$\var + 1 + 6 = d + 1$.
        Next, if the robot on~$\alpha_x^i$ is woken up, then this robot wakes up the robot on the unique neighbor of its vertex in the next layer.
        Afterwards, the robot returns to its initial position.
        This process continues for all robots on vertices of the path from~$\alpha_x^i$ to~$\omega_x^i$, which ensures that each such robot is woken up and arrives back at its initial position at the latest at time~$d+1$.

        What remains now is to wake up the robots on the variable and clause gadgets.
        This is initiated by the robot~$r^6_x$ (that is, the robot initially located at~$v_x$):
        If~$x$ is assigned TRUE, $r^6_x$ first wakes up the robot on~$x$, then moves back to~$v_x$, moves to~$\overline{x}$ and wakes up the robot there, and finally returns to its initial position at~$v_x$ at time~$\var + 4 < d+1$.
        If~$x$ is assigned FALSE, $r^6_x$ does the same while changing the order in which it visits~$x$ and~$\overline{x}$.
        Hence, for each literal~$\ell$ that is assigned TRUE, the robot initially located at~$\ell$ is woken up at time~$\var + 1$, and for each literal~$\ell'$ that is assigned FALSE, the robot initially located at~$\ell'$ is woken up at time~$\var + 3$.
        Moreover, each other robot initially located on a vertex of layer~$L_{\var + 1}$ is also woken up at time~$\var + 1$, that is, the vertices~$\{\alpha_x^i\mid x\in X, i\in [1,5]\}$.

        For each literal~$\ell$ that is assigned to TRUE, when the robot on~$\ell$ is woken up, it goes to $\ell''$.
        The robot on~$\ell'$ is woken up on the way, and also goes to~$\ell''$.
        This implies that at time~$\var + 3$, for each literal~$\ell$ that is assigned TRUE, there are three robots on~$\ell''$.
        As each literal occurs in at most three clauses, $\ell''$ has at most three neighbors in the next layer.
        Moreover, these vertices are exactly the vertices~$c\in C$ for which~$\ell$ is a literal of clause~$c$.
        The three robots on~$\ell''$ then wake up the (at most) three robots on adjacent vertices of~$\ell''$ at time~$\var + 4$ and afterwards return to their initial position at time at most~$\var + 4 + 3 = d + 1$.
        As the truth assignment is satisfying, this implies that for each clause~$c\in C$, the robot on vertex~$c$ is woken up at time~$\var + 4$.

        For the literals~$\ell$ that are assigned FALSE, the strategy is as follows:
        When the robot~$r$ on~$\ell$ is woken up at time~$\var + 3$, then~$r$ moves to~$\ell'$ and return back to~$\ell$ at time~$\var + 6 = d$. Moreover, the robot~$r'$ on~$\ell'$ that is in the process woken up at time~$\var + 4$ moves to~$\ell''$ and back to~$\ell$ to ensure that~$r'$ returns to its initial position at time~$\var + 6 = d$ and the robot on~$\ell''$ is woken up at time~$\var + 5$.

        It thus remains to describe how to wake up the vertices~$\{c',c''\mid c\in C\}$.
        This is essentially done in the same way as the vertices representing literals that are assigned FALSE:
        As for each clause~$c\in C$, the robot~$r$ on vertex~$c$ is woken up at time~$\var + 4$, there is enough time, that $r$ wakes up the robot~$r'$ on~$c'$ and return back to~$c$, and the robot~$r'$ that is woken up at time~$\var + 5$ can wake up the robot at~$c''$ and return back to~$c'$.

        Hence, given a satisfying truth assignment, we were able to describe a strategy which ensures that each robot is awake and back at its initial position at time at most~$d+1$.

        $(\Leftarrow)$
        Assume that there is a \FTRPa schedule with makespan~$d+1$.

        Consider the set~$R$ of vertices of the~$(\var+1)$-st layer for which the robots are woken up at time~$\var+1$.
        We will show that the literal-vertices contained in~$R$ model a satisfying truth assignment.
        That is, we will show that
        \begin{enumerate}
            \item\label{satisfying} for each clause~$c \in C$, there is at least one literal of~$c$ for which the corresponding vertex is in~$R$ and
            \item\label{eitheror} for each variable~$x$, $R$ contains either vertex~$x$ or vertex~$\overline{x}$.
        \end{enumerate}

        We start by showing Item~\ref{satisfying}.

        \begin{claim}\label{claim last layer early}
            For each vertex~$v \in L_d$, there is an ancestor~$u$ of~$v$ with~$u\in L_{\var + 1}$ such that the robot on~$u$ woken up at time~$\var + 1$.
        \end{claim}
        \begin{claimproof}
            Recall that for~$i\in [0,d]$, $L_i$ contains all vertices of distance~$i$ from the source~$s$, and that the ancestors of~$v$ in~$L_{\var+1}$ are all vertices of~$L_{\var+1}$ that are on some shortest~$(s,v)$-path.
            As robots can only move along the edges of~$E$ and the distance from~$s$ to~$v$ is~$d$, there is a shortest~$(s,v)$-path in~$G$ along which all edges are traversed in consecutive times by robots.
            In particular, this implies that for each~$i\in [0,d-1]$, there is an ancestor of~$v$ in~$L_i$ that is woken up at time~$i$.
            Hence, the claim follows as a corollary.
        \end{claimproof}

        Note that this implies that~$R$ contains all the vertices~$\alpha_x^i$ with~$x\in X$ and~$i\in [1,5]$, as~$\omega_x^i$ is a vertex of~$L_{d}$, and~$\alpha_x^i$ is the only ancestor of~$\omega_x^i$ in~$L_{\var+1}$.

        \begin{claim}
            For each clause~$c\in C$, $R$ contains at least one vertex representing a literal of~$c$, that is, a vertex of~$\{\ell\mid \ell \in c\}$.
        \end{claim}

        \begin{claimproof}
            By construction, for clause~$c$, there is a vertex~$c''$ in~$L_d$.
            Thus, by~\Cref{claim last layer early}, there is an ancestor~$w\in L_{\var + 1}$ that is woken up at time~$\var+1$, that is, for which~$R$ contains~$w$.
            Furthermore, by construction, the only ancestor of~$c''$ in~$L_{d-2}$ is~$c$ and the only neighbors of~$c$ in~$L_{d-3}$ are the vertices~$\ell''$ for each literal~$\ell$ of clause~$c$.
            Moreover, for each such vertex~$\ell''$ which is an ancestor of~$c''$ in~$L_{d-3}$, vertex~$\ell$ is the only ancestor of~$\ell''$ in~$L_{d-5} = L_{\var+1}$.
            That is, the literals of~$c$ are exactly the ancestors of~$c''$ in~$L_{\var+1}$.
            Thus, the claim follows.
        \end{claimproof}

        This thus immediately implies that the literals that are contained in~$R$ satisfy all clauses, hence proving~Item~\ref{satisfying}.

        Next, we show that for each variable~$x$, $R$ contains either~$x$ or~$\overline{x}$.
        By the above, we immediately get that~$R$ contains at least one of~$x$ or~$\overline{x}$, as the formula contains the clause~$(x\lor \overline{x})$.
        It thus remains to show that~$R$ does not contain both~$x$ and~$\overline{x}$.

        \begin{claim}\label{claim:eitheror}
            For each variable~$x$, $R$ contains either~$x$ or~$\overline{x}$.
        \end{claim}

        \begin{claimproof}
            Assume towards a contradiction that there is a variable~$x$, such that~$R$ contains both~$x$ and~$\overline{x}$.
            Then, $R$ contains at least~$6\cdot |X| + 1$ vertices, as (i)~$R$ contains all~$5\cdot |X|$ vertices of the~$(\var + 1)$-st layer that do not represent literals (that is, the vertices~$\{\alpha_x^i\mid x\in X, 1 \leq i \leq 5\}$) and (ii)~$R$ contains for each variable~$y\neq x$ at least one of~$y$ and~$\overline{y}$, and (iii)~$R$ contains both~$x$ and~$\overline{x}$.
            Note that all robots that woke up vertices of~$R$ are originally from a layer~$L_i$ with~$0\leq i \leq \var$.
            Hence, there are at least $|R| \geq 6\cdot |X| + 1$~many vertices of~$V_{\leq \var} := \bigcup_{i=0}^\var L_i$ that do not contain a robot at the end of time~$\var+1$.
            As we have a \FTRPa schedule, at least $|R|$~robots that are not on vertices of~$V_{\leq\var}$ at the end of time~$\var+1$ have to end on pairwise distinct vertices of~$V_{\leq\var}$ by the end of time~$d +1= \var + 7$.
            By construction, each of the layers from~$L_{\var - 5}$ to~$L_\var$ contain exactly $|X|$~vertices each.
            As pairwise distinct vertices need to be reached, at least one robot that is on a vertex of~$V\setminus V_{\leq \var}$ at the end of time~$\var +1$ thus has to end on at least one vertex of~$L_i$ with~$i < \var - 5$ by the end of time~$d +1= \var + 7$.
            This is not possible, as the distance between the respective layers is larger than~$6$.
            This contradicts the assumption that~$R$ contains both~$x$ and~$\overline{x}$.
        \end{claimproof}

        Thus also Item~\ref{eitheror} holds.
        Consequently, the literals of~$R$ model a satisfying truth assignment, which implies that the formula is satisfiable.

        By the fact that the number of vertices of the graph is linear in the size of the formula, this implies that neither of our problems can be solved in time $2^{o(n)}$, as the considered version of \SAT cannot be solved in $2^{o(|F|)}$ time, unless the ETH fails, thanks to the linear reduction of~\cite{Tovey84}.
    \end{proof}
}



\begin{restatable}{corollary}{NPcorollary}\label{cor:NP-corollary}
    \FTRPid and \FTRPa are NP-hard in metric spaces and cannot be solved in $2^{o(n)}$~time, unless the ETH fails.
\end{restatable}

\newcommand\NPmetricproof{
    \begin{proof}
        We reduce from the respective problem itself on undirected and unweighted graphs.
        Due to the previous theorem, this is NP-hard.

        Let~$(G=(V,E),t)$ be a respective instance.
        We define a distance function~$\omega \colon V\times V\to \mathbb{N}$ as follows:
        For each vertex~$v\in V$, we set~$\omega(v,v) := 0$.
        For each edge~$\{u,v\}\in E$, we set~$\omega(u,v) := \omega(v,u) := 1$.
        For each non-edge~$\{u,v\}\in \binom{V}{2}\setminus E$, we set~$\omega(u,v) := \omega(v,u) := |uv|$, where~$|uv|$ denotes the length of any shortest~$(u,v)$-path in~$G$.

        Clearly, $\omega$ is idempotent, symmetric, and fulfills the triangle inequality. 

        We argue that the optimal makespan of any $\FTRPid$ (resp. \FTRPa) strategy for~$G$ equals the optimal makespan of any $\FTRPid$ (resp. $\FTRPa$) strategy for~$(V, \omega)$.

        As the endpoints of an edge of~$G$ have a pairwise distance of~$1$ under~$\omega$, each strategy for~$G$ can also be used for~$(V, \omega)$ and achieves the same makespan.
        Thus, the optimal makespan for~$(V,\omega)$ does not exceed the optimal makespan for~$G$.

        For the other direction, it suffices to show that there is an optimal strategy in which robots only go from~$u$ to~$v$ if~$\{u,v\}$ is an edge of~$G$.
        To this end, observe that if a strategy lets a robot~$r$ go from~$u$ directly to~$v$ if~$\{u,v\}$ is not an edge of~$G$, then we can replace this move by instead sending~$r$ along a shortest~$(u,v)$-path~$P$ in~$G$.
        By definition, the length of~$P$ is equal to~$\omega(u,v)$.
        Hence, under the new strategy, robot~$r$ arrives at the same time as in the original strategy.
        That is, the new strategy has the same makespan as the old strategy, but uses less direct movements between vertices that are not adjacent in~$G$.
        By an inductive argument, this implies that there is an optimal wake-up strategy for~$(V,\omega)$, for which robots only move along the edges of~$G$.
        That is, the optimal makespan for~$G$ does not exceed the optimal makespan for~$(V,\omega)$.

        Hence, the optimal makespan for~$G$ is equal to the optimal makespan for~$(V,\omega)$.
        As~$\omega$ is a metric, we conclude the proof.
    \end{proof}
} 

\section{Optimal Algorithms}\label{sec:optimal}

In this part we develop algorithms for \FTP and its variants. They return an optimal solution, that is, the optimal makespan, assuming any non-negative distance function between robots' initial positions. In particular, they apply to metric spaces as well as weighted graphs, even if the triangle inequality does not hold.

It is clear that any brute-force algorithm needs to consider at least $n!$ solutions, 
that is at least as many solutions as the path-TSP\footnote{In the \emph{Path Traveling Salesperson Problem}, one asks for a path of minimum length connecting all given points.} and the classical TSP, since these latter problems are valid solutions for the \FTP and \FTRP. Actually, this is even worse when considering all possible wake-up trees (that can be seen as labeled and binary trees), as this number raises up to $2^{\Omega(n)} \cdot n!$.
In the following, we assume that the given distance function can be evaluated in constant time.

\begin{restatable}{theorem}{optimalalgo}%
    \label{th:optimal-algo}
    For \FTP, \FTRPid, and \FTRPa with any non-negative distance function, there are deterministic algorithms returning the optimal solution in time complexity $O(3^n \cdot n^2)$, $O(3^n \cdot n^3)$, and $O(9^n \cdot n^{3/2})$ respectively, and space complexity $O( 2^n \cdot n)$, $O(2^n \cdot n^2)$, and $O(4^n \cdot n)$ respectively, where $n$ is the number of sleeping robots.
\end{restatable}

We stress that from our hardness results in \cref{th:NPhard}, no $2^{o(n)}$-time algorithms can exist for these problems. The proofs are provided in \cref{app:optimal-algo} in the Appendix.

Our algorithms are inspired by the famous Held--Karp's algorithm, independently discovered by Bellman in the 1960s. We use a recursive approach which for \FTP is actually quite close to the original Bellman--Held--Karp's formulation of the path-TSP sub-problem decomposition of TSP. Then, we use a dynamic programming approach with exponential space. For the implementation of \FTP and \FTRPid, we can use a relatively simple pair of tables of respectively $2^n \times n$ and $2^n \times n^2$ cells, where each cell is assumed to be large enough to store a distance. Whereas for \FTRPa, we analyze a standard memoization technique which leads, 
to a roughly $9^n$-time algorithm.


\section{Conclusion}%
\label{sec:contwp vclu}

We introduced two versions of the Freeze-Tag-with-Return Problem, a natural variant of the Freeze-Tag Problem.
We showed that both versions are NP-hard when the underlying space is an unweighted undirected graph.
However, it remains open whether the problems are NP-hard in other metric spaces, like $(\mathbb{R}^d,\ell_p)$ spaces ($d=2$ and $p=2$, the Euclidean plane).

Regarding the presented lower and upper bounds for the optimal makespan, we conjecture that they are not optimal for large instances. We found a distribution of positions (for any $n\geq 5$) with an optimal makespan of $4.67$. We still have a gap between this lower bound and the two upper bounds of $4.89 + O(1/\sqrt{n}\,)$ for arbitrary distribution and of $2 + 2\sqrt{2} \approx 4.83$ for convex distribution.

In general, we conjecture that the optimal makespan of worst-case instances of both introduced problems, i.e., the wake-up return constants, never exceeds~$2+2\sqrt{2}$. 
In other words,

\begin{conjecture}
    $\OPTid = \OPTa = 2+2\sqrt{2}$.
\end{conjecture}

It is also interesting to further analyze the connection between the worst-case makespan between \FTP and \FTRP for any fixed distribution of sleeping robots~$\P$. We gave examples, where this difference can be larger than~$1.732$ but showed that it never exceeds~$1.959$. The latter was shown via an algorithm 
that takes any solution for \FTP for the set~$\P$ and computes a solution for \FTRP for the set~$\P$ where the makespan does only increase by at most~$1.959$. There are two main directions to improve in this area: 1) reduce the difference between bounds and 2) get a simpler and polynomial algorithm to achieve a solution for \FTRP from a solution for \FTP. 

Finally, the definition of the \FTRP might inspire the introduction of other variants of \FTP, like for example, the task where all robots have to gather at the same point after the wake-up process is finished.




\newpage

\ifarxiv
\bibliographystyle{my_alpha_doi}
\def\MYMARGIN{17mm}
%
%

\makeatletter
\ifdefined\MYMARGIN\else\def\MYMARGIN{18.5mm}\fi
\renewenvironment{thebibliography}[1]
  {\if@noskipsec \leavevmode \fi
   \par
   \@tempskipa-3.5ex \@plus -1ex \@minus -.2ex\relax
   \@afterindenttrue
   \@tempskipa -\@tempskipa \@afterindentfalse
   \if@nobreak
     \everypar{}%
   \else
     \addpenalty\@secpenalty\addvspace\@tempskipa
   \fi
   \noindent
   \rlap{\color{lipicsLineGray}\vrule\@width\textwidth\@height1\p@}%
   \hspace*{7mm}\fboxsep1.5mm\colorbox[rgb]{1,1,1}{\raisebox{-0.4ex}{%
     \normalsize\sffamily\bfseries\refname}}%
   \@xsect{1ex \@plus.2ex}%
   \list{\@biblabel{\@arabic\c@enumiv}}%
        {\leftmargin\MYMARGIN
         \labelsep\leftmargin
         \settowidth\labelwidth{\@biblabel{#1}}%
         \advance\labelsep-\labelwidth
         \usecounter{enumiv}%
         \let\p@enumiv\@empty
         \renewcommand\theenumiv{\@arabic\c@enumiv}}%
   \fontsize{9}{12}\selectfont
   \sloppy
   \clubpenalty4000
   \@clubpenalty \clubpenalty
   \widowpenalty4000%
   \sfcode`\.\@m\protected@write\@auxout{}{\string\gdef\string\@pageNumberStartBibliography{\thepage}}}
\makeatother

\else
\bibliographystyle{plainurl} 
\fi

\bibliography{biblio}

\newpage
\appendix

\section{Extra notations and reminders}\label{appendix:notations}

In the various proofs, we use the following notations.
Some of them come from previous articles about FTP.

We denote by $\P^* = \P \setminus\{\pos_0\}$ the set of positions of sleeping robots.

In the unit disk, we often use the \emph{length of a chord} of angle $\alpha$, denoted $\chord(\alpha) = 2\sin(\alpha/2)$.

We define a \emph{crown} of angle $\alpha$ and width $w$ as a region containing all points between distance $1-w$ and $1$ from the origin of the Euclidean plane and whose angular coordinates are within a cone of angle $\alpha$ centered at $\pos_0$.

Different solutions of $\FTP$ are provided in \cite{BCGH24,BGHO24} depending on the shape of regions containing sleeping robots:

\begin{itemize}

    \item For crowns of angle $\alpha$ and width $w$, assuming that the initial awake robot lies on a line segment boundary of the crown, Algorithm $\crown$ returns a wake up tree of weighted depth $\alpha+w(1+\varphi)$.

    \item For cones of angle $\alpha$ and of radius $R$, assuming that the initial awake robot is on the cone's apex, 
    Algorithm $\splitconestrategy$ returns a wake-up tree of weighted depth $R(1+\varphi)$.

\end{itemize}


\section{Lower and upper bounds: detailed proofs}

\subsection{Lower bounds}\label{appendix:lowerbound-nsmall}

For small values of $n$, when the points are regularly distributed on the unit circle, this already gives us interesting lower bounds (tight for $n=3$ and $n=4$) (see \cref{fig:convid} for \FTRPid and \cref{fig:conva} for \FTRPa).
In these examples, we use bold edges to show the critical path in the wake-up trees.

\begin{figure}[htbp!]
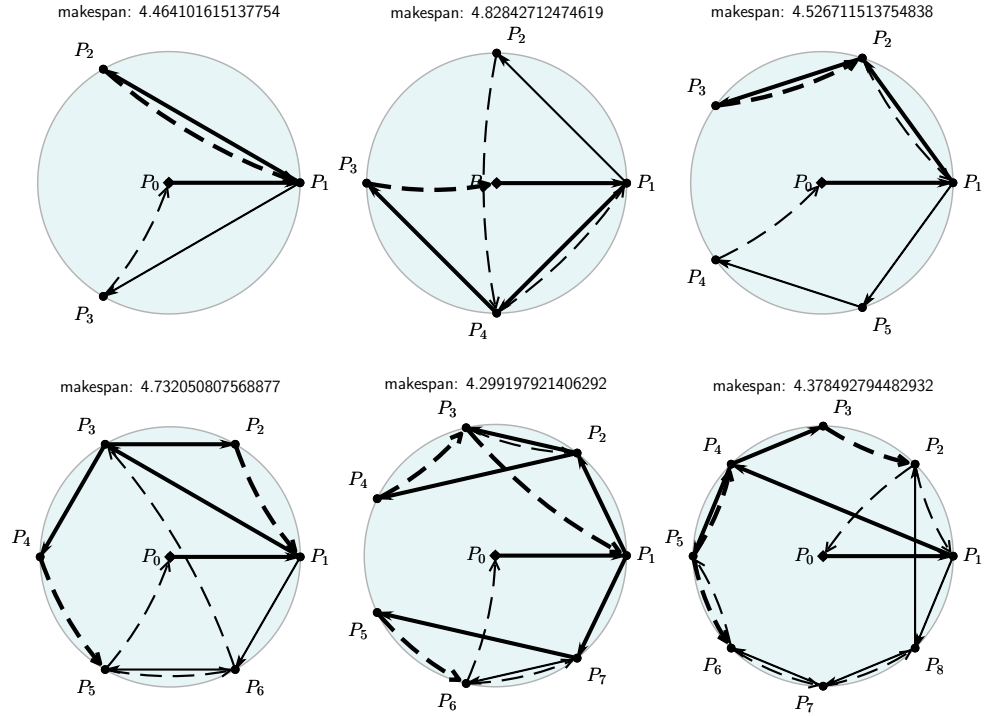

    \centering
    \begin{subfigure}[b]{0.3\linewidth}
        \incfig{\FigRegularOptBiIdThree}{1}
    \end{subfigure}
    \begin{subfigure}[b]{0.3\linewidth}
        \incfig{\FigRegularOptBiIdFour}{1}
    \end{subfigure}
    \begin{subfigure}[b]{0.3\linewidth}
        \incfig{\FigRegularOptBiIdFive}{1}
    \end{subfigure}
    \hfill
    \begin{subfigure}[b]{0.3\linewidth}
        \incfig{\FigRegularOptBiIdSix}{1}
    \end{subfigure}
    \begin{subfigure}[b]{0.3\linewidth}
        \incfig{\FigRegularOptBiIdSeven}{1}
    \end{subfigure}
    \begin{subfigure}[b]{0.3\linewidth}
        \incfig{\FigRegularOptBiIdEight}{1}
    \end{subfigure}
    \caption{Optimal wake-up trees for small regular convex configurations for \FTRPid.}%
    \label{fig:convid}
\end{figure}

\begin{figure}[htbp!]
    \centering
    \begin{subfigure}[b]{0.3\linewidth}
        \incfig{\FigRegularOptCovThree}{1}
    \end{subfigure}
    \begin{subfigure}[b]{0.3\linewidth}
        \incfig{\FigRegularOptCovFour}{1}
    \end{subfigure}
    \begin{subfigure}[b]{0.3\linewidth}
        \incfig{\FigRegularOptCovFive}{1}
    \end{subfigure}
    \hfill
    \begin{subfigure}[b]{0.3\linewidth}
        \incfig{\FigRegularOptCovSix}{1}
    \end{subfigure}
    \begin{subfigure}[b]{0.3\linewidth}
        \incfig{\FigRegularOptCovSeven}{1}
    \end{subfigure}
    \begin{subfigure}[b]{0.3\linewidth}
        \incfig{\FigRegularOptCovEight}{1}
    \end{subfigure}
    \caption{Optimal wake-up trees for small regular convex configurations for \FTRPa.}%
    \label{fig:conva}
\end{figure}

\lowerboundnsmall*


\begin{proof}
    The lower bounds for $n=3$ and $n=4$ are based on regular distribution of the sleeping robots on the unit circle.
    In each case, we argue about the weighted depth of possible wake-up trees.
    For $n < 5$, wake-up trees have either one or two leaves:
    Since the root has only one child, and each node has at most two child, we need to have at least $3$ non-leaf nodes to get more than two leaves in the wake-up tree.

    \begin{claim}
        For $n=3$, $\OPTidn(3)\geq \OPTan(3) \geq 1+2\sqrt{3} \approx 4.46$
    \end{claim}
    \begin{claimproof}
        First note that the pairwise distance of the three sleeping robots' positions $\pos_1$, $\pos_2$ and $\pos_3$ is $\chord(2\pi/3)=\sqrt{3}$.
        The smallest makespan achieved by a wake-up tree with only one leaf (and one return arc) corresponds to the circuit $\pos_0,\pos_1,\pos_2,\pos_3,\pos_0$.
        The makespan is $2+2\sqrt{3}$.
        We can have a smaller makespan with a wake-up tree with two leaves.
        The optimal wake-up tree for the $\FTP$ takes $1+\sqrt{3}$: one unit for $r_0$ to reach $r_1$ and $\sqrt{3}$ to wake up $r_2$ and $r_3$ in parallel.
        $r_0$ (resp. $r_1$) returns to its initial position in time $1$ (resp. $\sqrt{3}$).
        Thus we have a makespan of $1+2\sqrt{3}$.
        We can not do better since the two awake robots reaching the two leaves have to return to different initial positions and one of them is at distance at least $\sqrt{3}$.
        The first hop in the wake-up tree has length $1$ whereas the second hop has necessary length $\sqrt{3}$.
    \end{claimproof}

    \begin{claim}
        For $n=4$, $\OPTidn(4)\geq \OPTan(4) \geq  2+2\sqrt{2} \approx 4.824,$
    \end{claim}
    \begin{claimproof}
        We show that the smallest makespan is achieved by a wake-up tree with only one leaf that takes a time $2+2\sqrt{2}$.
        Assume that we have two leaves in a wake-up tree with smallest makespan.
        The unweighted depth is at least $3$.
        Consider the root-leaf path $\pos_0, \pos_i, \pos_j, \pos_k$ of a leaf $\pos_k$ reached in $3$ hops.
        The first hop $\pos_0, \pos_i$ has length $1$ and the next two hops have length at least $\sqrt{2}$ since the pairwise distance between sleeping robots is at least $\sqrt{2}$.
        The robot that moved to $\pos_k$ then has to return to an initial position, which is at a distance of at least $1$ from $\pos_k$.
        Thus the length of the trajectory of this robot is at least $2+2\sqrt{2}$.
    \end{claimproof}

    \begin{figure}[htbp!]
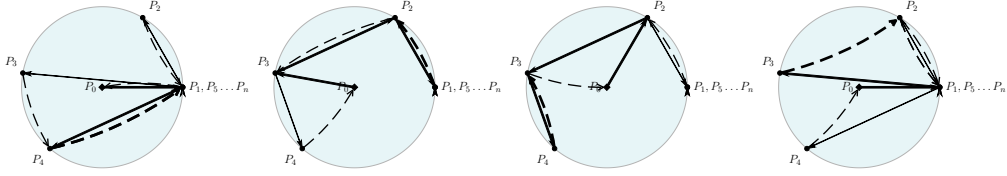

        \centering
        \begin{subfigure}[b]{0.23\textwidth}
            \incfig{\FigWorstCovCA}{1}
        \end{subfigure}
        \begin{subfigure}[b]{0.23\textwidth}
            \incfig{\FigWorstCovCB}{1}
        \end{subfigure}
        \begin{subfigure}[b]{0.23\textwidth}
            \incfig{\FigWorstCovCC}{1}
        \end{subfigure}
        \begin{subfigure}[b]{0.23\textwidth}
            \incfig{\FigWorstCovCD}{1}
        \end{subfigure}
        \caption{Positions are such these four wake-up trees are optimal. Critical paths are bold lines. Observe in the last wake-up tree, that one $r_1$ wakes up $r_2$ and then they both go back to $\pos_1=\pos_5$.}%
        \label{fig:optworstcov}
    \end{figure}

    \begin{claim}
        For any $n \geq 5$, $\OPTan(n) \geq 1 + \frac{1}{4} (7^{1/4} + 7^{3/4})\sqrt{6} > 4.631$
    \end{claim}
    \begin{claimproof}
        We spread the $n$ sleeping robots on $4$ sites: $n-3$ sleeping robots are located on the unit circle with the same angle $\alpha_1 = 0$ ($\pos_1=\pos_5=\ldots=\pos_{n}$) whereas the angles of the $3$ remaining robots are
        \begin{align*}
            \alpha_2 & = 2\arctan\pare{\frac{7^{1/4} \sqrt{6}}{3\sqrt{7}-1}},                               \\
            \alpha_3 & = \alpha_2+ 2\arctan\pare{\frac{7 ^{1/4} \sqrt{6}}{(\sqrt{7} - 1)(4 \sqrt{7} - 6)}}, \\
            \alpha_4 & =\alpha_2 + \alpha_3.
        \end{align*}

        With these angles, there are four optimal wake-up trees, as  presented in \cref{fig:optworstcov}.
        Their makespan are given by the following:
        \begin{eqnarray*}
            &=& |\pos_0\pos_1|+2|\pos_1\pos_3| \\
            &=& |\pos_0\pos_3|+|\pos_3\pos_2| + 2|\pos_2\pos_1| \\
            &=& |\pos_0\pos_2|+|\pos_2\pos_3| + 2|\pos_3\pos_4| \\
            &=& |\pos_0\pos_1|+|\pos_1\pos_3| + |\pos_3\pos_2|\\
            &=& 1 +\frac{1}{4}(7^{1/4} + 7^{3/4})\sqrt{6} \approx 4.631.
        \end{eqnarray*}
        The algorithm presented in \cref{sec:optimal} explores all these wake-up trees and check that all other wake-up trees are sub-optimal.

        \begin{figure}[htbp!]
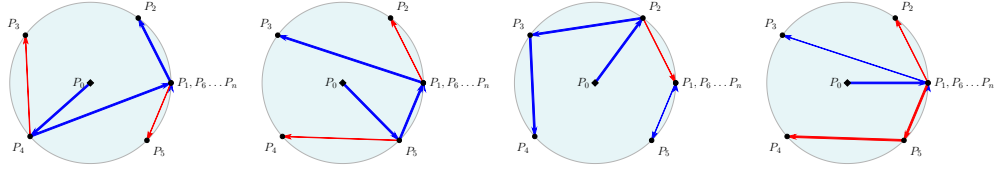

            \centering
            \begin{subfigure}[b]{0.23\textwidth}
                \incfig{\FigWorstIdCA}{1}
            \end{subfigure}
            \begin{subfigure}[b]{0.23\textwidth}
                \incfig{\FigWorstIdCB}{1}
            \end{subfigure}
            \begin{subfigure}[b]{0.23\textwidth}
                \incfig{\FigWorstIdCC}{1}
            \end{subfigure}
            \begin{subfigure}[b]{0.23\textwidth}
                \incfig{\FigWorstIdCD}{1}
            \end{subfigure}
            \caption{Positions are such these four wake-up trees are optimal.}%
            \label{fig:optworstid}
        \end{figure}
    \end{claimproof}
    \begin{claim}
        For any $n \geq 5$, $\OPTidn(n) > 4.672.$
    \end{claim}
    \begin{claimproof}
        We spread the $n$ sleeping robots on $5$ sites: $n-4$ sleeping robots are located on the unit circle with the same angle $\alpha_1 = 0$ ($\pos_1=\pos_6=\ldots=\pos_{n}$) whereas the angles of the $4$ remaining robots are
        \begin{align*}
            \alpha_2 & \approx 0.935362,               \\
            \alpha_3 & \approx \alpha_2 + 1.572573134, \\
            \alpha_4 & \approx \alpha_3+  1.359792,    \\
            \alpha_5 & \approx \alpha_4 + 1.622400.
        \end{align*}

        With these angles, there are four optimal wake-up trees, as  presented in \cref{fig:optworstid}.
        Their makespans are given by the following trajectories:
        \begin{eqnarray*}
            &=& |\pos_0\pos_4|+|\pos_4\pos_1| + |\pos_1\pos_0|, \\
            &=& |\pos_0\pos_4|+|\pos_5\pos_1| + |\pos_1\pos_3| + |\pos_1\pos_0|,\\
            &=& |\pos_0\pos_2|+|\pos_2\pos_3| + |\pos_3\pos_4| + |\pos_4\pos_0|,\\
            &=& |\pos_0\pos_1|+|\pos_1\pos_3|+|\pos_3\pos_1|.\\
        \end{eqnarray*}
        In each cases, the makespan is greater than $4.672$.
        The algorithm presented in \cref{sec:optimal} explores all other wake-up trees and check that they have a greater makespan.
    \end{claimproof}
\end{proof}

\subsection{Upper bound for convex positions}%
\label{appendix:UBConvex}


We start by stating and proving technical lemmas (\cref{lem:disk projection}, \cref{lem:linear} and \cref{lem:sum of chords}), then we describe the dedicated  algorithm \convex, and we show it achieves a makespan of $2 + 2\sqrt{2}$ for the \FTRPid (\cref{th:convex_UB}).
Furthermore we extend \cref{th:convex_UB} to \cref{cor:upperboundSmall} by providing an analysis for special cases where $n \in \set{3,4}$. 

\subparagraph{Projection on the circle.}

When it comes to upper bounding distances between a pair of points in a convex set, it comes in handy to assume that sleeping robots lie on the border of the unit disk centered in $\pos_0$.
To this end, we use the projection $\psi$ (see~\cite[proof of Theorem~2]{BGHO24}) of $\P^*$ on $\mathcal{C}$ as follows:
For each $\pos_i \in \P^*$, $\psi(\pos_i) \in \mathcal{C}$ is the intersection of the half-line perpendicular to $\pos_{i-1}\pos_i$ emanating from $\pos_i$ and going outside from the convex hull of $\P$.\nicolasB{I changed $\varphi$ by $\psi$ to avoid confusion with the golden ratio.}


\begin{restatable}{lemma}{Projection}\label{lem:disk projection}
    Let $\P^*$ be a set of points in convex position within the unit disk, and let $\mathcal{C}$ be circle of unitary radius.
    The projection $\psi : \P^* \mapsto \mathcal{C}$ is such that for any two points $\pos_i, \pos_j$, $|\psi(\pos_i)\psi(\pos_j)| \ge |\pos_i\pos_j|$.
\end{restatable}

\begin{proof}
    Let $\pos_i, \pos_j$ be two points in a convex set.
    We start by a technical claim.
    Let $\mathcal{H}_i$ (\emph{resp.} $\mathcal{H}_j$) be the half-space bounded by the perpendicular line $\mathcal{L}_i$ (\emph{resp.} $\mathcal{L}_j$) to the segment $\pos_i\pos_j$ going through $\pos_i$ (\emph{resp.} $\pos_j$) that does not contain $\pos_j$ (\emph{resp.} $\pos_i$).
    \begin{claim}\label{half-space claim}
        Let $\pos_i'$ be any point in $\mathcal{H}_i$ (\emph{resp.} $\pos_j' \in \mathcal{H}_j$). We have $|\pos_i'\pos_j'| \geq |\pos_i\pos_j|$.
    \end{claim}

    \begin{claimproof}
        Let $\pos_i'' \in \mathcal{L}_i$ and $\pos_j'' \in \mathcal{L}_j$ be the intersection points of the segment $\pos_i'\pos_j'$.
        Since $\mathcal{L}_i$ and $\mathcal{L}_j$ are parallel, the distance between $\pos_i''$ and $\pos_j''$ is at least $|\pos_i\pos_j|$.
        Thus $|\pos_i'\pos_j'| \geq |\pos_i''\pos_j''| \geq |\pos_i\pos_j|$.
    \end{claimproof}

    \begin{claim}\label{projection claim}
        For each pair of points $\pos_i, \pos_j \in \P$, $|\psi(\pos_i)\psi(\pos_j)| \geq |\pos_i\pos_j|$.
    \end{claim}

    \begin{claimproof}
        Let $\pos_i' = \psi(\pos_i)$.
        Given $\pos_j$, let $\theta = \angle \pos_j\pos_i\pos_{i-1}$.
        For convenience, let $A$ be any point on the perpendicular line to $\pos_j\pos_i$ going through $\pos_i$ such that $\angle \pos_j\pos_i A = \frac{\pi}{2}$.
        Since $\P$ is a convex point set, $\theta \in [0, \pi]$.
        We have $\angle A \pos_i \psi(\pos_i) = \angle \pos_j \pos_i \psi(\pos_i) - \angle \pos_j \pos_i A = \angle \pos_j \pos_i \pos_{i-1} + \angle \pos_{i-1} \pos_i \psi(\pos_i) + \angle \pos_j \pos_i A = \theta + \frac{\pi}{2} - \frac{\pi}{2} = \theta \in [0, \pi]$.
        It turns out that $\psi(\pos_j)$ is not in the same half-space bounded by the line $(A, \pos_i)$ containing $\pos_j$.
        Similarly, $\psi(\pos_j)$ is not in the half-space bounded by the parallel of $(A, \pos_i)$ going through $\pos_j$ containing $\pos_i$.
        We then use the previous Claim~\ref{projection claim} to get that $|\psi(\pos_i)\psi(\pos_j)| \geq |\pos_i\pos_j|$.
    \end{claimproof}
\end{proof}

\subparagraph{Cyclic ordering.}

In convex configuration, the sleeping robots can be \emph{cyclically ordered} 
by increasing angular distance from a fixed point.
We say that two robots are \emph{consecutive} if they are consecutive within any cyclic order.

\subparagraph{\linear algorithm.}

We describe the subroutine \linear that is used for awakening a set of consecutive robots.
Given a set of $n$ initial positions located on an arc of length $\alpha$ that are cyclically ordered as $\pos_1, \pos_2, \ldots, \pos_{n}$, where $\pos_1$ is supposed initially awake, the algorithm goes as follow:
For $1 \leq i \leq n-1$, the robot from $\pos_i$ wakes up the robot from $\pos_{i+1}$ and then returns to its position.

\begin{lemma}\label{lem:linear}
    Let $\Tlinear(\alpha, \beta)$ be the makespan of \linear applied on robots located on an arc of length $\alpha \leq \pi$ with a maximal angular distance less than $\beta$.
    If $\beta > 2\pi/3, \Tlinear(\alpha, \beta) \leq\alpha-2 \pi/3 + 2 \sqrt{3}$, and $\Tlinear(\alpha, \beta) \leq \alpha-\beta+ 2\chord(\beta)$ otherwise.
\end{lemma}

\begin{proof}
    Assume that the angular distance between consecutive robots positions are $\alpha_1,\alpha_2,\ldots,\alpha_k$ with $\alpha=\sum_{i=1}^k \alpha_i$.
    Note that the wake-up time for the $j$-th sleeping robot is $\sum_{i=1}^{j<k} \chord(\alpha_i)$ and its \completionTime{} is $(\sum_{i=1}^{j<k} \chord(\alpha_i))+2\chord(\alpha_{j+1}) < \sum_{i=1}^{j<k} \alpha_i + 2 \chord(\alpha_{j+1}) \leq \alpha - \alpha_{j+1}+2 \chord(\alpha_{j+1})$.
    Function $f(x)=\alpha-x+2\chord(x)$ is increasing up to $x=2\pi/3$ and then is decreasing.
    Thus, we have $\Tlinear(\alpha, \beta) \leq \alpha-2 \pi/3 + 2 \sqrt{3}$ and for $\beta \leq 2\pi/3, \Tlinear(\alpha, \beta) \leq \alpha-\beta+ 2\chord(\beta)$.
\end{proof}

\begin{figure}[htbp!]
    \centering
    \includegraphics[width=\textwidth/3]{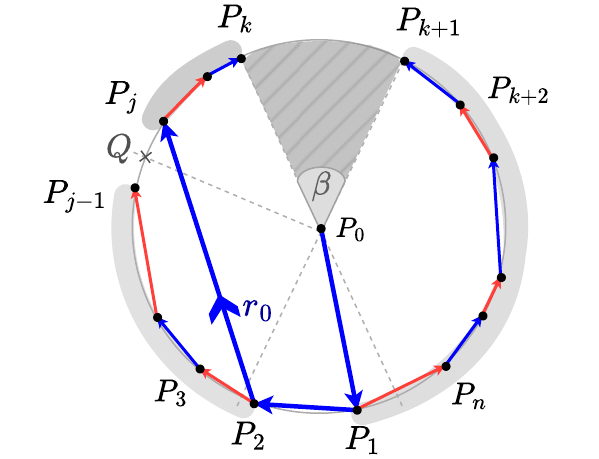}
    \caption{Bi-colored wake-up tree resulting from the Convex Algorithm.
        Highlighted arcs are woken up using \linear}%
    \label{fig:convex-algo-appendix}
\end{figure}

We are now ready to focus on the main statement of the section.
Our proof relies on the analysis of the makespan of a wake up strategy computed by the following algorithm, \convex.

\paragraph*{\convex Algorithm.}

Let $\pos$ and $\pos'$ be two consecutive position in $\P$ such that the angular distance from $\pos$ to $\pos'$, denoted $\beta$, is maximal.
We denote by $r_1$ the robot whose initial position $\pos_1$ is the closest to the bisector of the angle $\beta$ (this angle being empty by definition, $\pos_1$ must sit on the arc diametrically opposed to the arc $(\pos \pos')$).
In the latter we consider a cyclic ordering of the positions in $\P$ starting from $\pos_1$, where the $i$-th position is denoted $\pos_i$ and the corresponding robot $r_i$.
We denote by $k$ the integer such that $\pos = \pos_k$ and $\pos' = \pos_{k+1}$.

The location of $\pos_1$ partitions the arc $(\pos_{k+1}\pos_k)$ into two sectors $(\pos_1\pos_k)$ and $(\pos_{k+1}\pos_1)$.
In the following we assume that $|\pos_1\pos_k| \geq |\pos_{k+1}\pos_1|$.
If it is not the case, the strategy is reflected with respect to the axis bisecting $\beta$.

The robot initially at $\pos_0$ wakes up $r_1$ and moves toward the next robot in the direction of $\pos_k$, denoted $\pos_2$, while $r_1$ executes \linear on the arc $(\pos_1\pos_{k+1})$.
We then subdivide $(\pos_2\pos_k)$ into two sub-arcs $(\pos_2\textsc{Q})$ and $(\textsc{Q}\pos_k)$, where $\textsc{Q}$ is a geometric point at a angular distance $\alpha_q=1+2\sqrt{2} + 2 \pi/3-2\sqrt{3} \approx 2.45$ from $\pos_1$ in the direction of $\pos_k$.
In arc $(\pos_2\textsc{Q})$, \linear is computed starting from robot $r_2$, while the second robot, $r_0$, wakes up the first, thus closest, robot $r_j$ in arc $(\textsc{Q}\pos_k)$.
Finally, $r_0$ returns to the center of the disk, while $r_j$ runs \linear in the arc $(\pos_j\pos_k)$.
Note that if $\alpha_q$ is larger than the length of the arc $(\pos_1,\pos_k)$ $r_0$ directly goes to its initial position after having woken up $r_2$.

\convexUB*

\begin{proof}

    In our proof we assume that robots lies on the circle of unitary radius centered at $\pos_0$.
    If it is not the case, by \cref{lem:disk projection} we can project points of $\P\setminus \{\pos_0\}$ onto the circle while ensuring that pairwise distances between points are either preserved or increased.
    Hence we can use \convex to compute a strategy for the projected points and apply it to the original instance, the upper bound will still hold.

    We assume that $(\pos_k, \pos_1) \geq (\pos_1,\pos_{k+1})$, otherwise, and as mentioned in the algorithm description, the strategy and the analysis are mirrored along the bisector of the angle $\beta$.

    We have to bound the makespan for four types of robots:
    \begin{enumerate}
        \item The trajectory of the initial awake robot $r_0$ can be represented by the sequence of five points $(\pos_0, \pos_1, \pos_2, \pos_j, \pos_0)$.
              It makes at most $4$ hops since it may happen that $\pos_2$ or $\pos_j$ does not exist.
              In this case, $r_0$ can do $3$ hops and even $2$ hops if $\pos_1=\pos_k$.
              In any case, the first and the last hops are of length $1$: $|\pos_0\pos_1| + |\pos_j\pos_0| = 2$.
              By construction, the arc $(\pos_1\pos_k)$ has a length smaller or equal to $\pi$, and $|\pos_1\pos_2| + |\pos_2\pos_j| \leq |\pos_1\pos_k|$.
              If we have $4$ hops (or $3$), we can bound the length of $|\pos_1\pos_j|$ by \cref{lem:sum of chords}: $|\pos_1\pos_j|$ is smaller than $2 \chord(\pi/2) = 4 \sin(\pi/4) = 2 \sqrt{2}$. \nicolasB{This Lemma is used before being stated, maybe we should state it before ?}

        \item Robots $\set{r_1, r_{k+1}, r_{k+2}, \dots,r_n}$ lying on the arc $(\pos_{k+1},\pos_1)$.

              The makespan of this arc relies on the analysis of the \linear algorithm.
              Since $r_1$ is awakened at time $1$, we get that the makespan of the arc is $1 + \Tlinear(\alpha', \beta')$ where $\alpha'$ is the length of $(\pos_{k+1}, \pos_1)$ and $\beta'$ is the maximal angular distance between two robots on the arc.
              On one hand, we have $\beta' \leq \alpha' \leq  \pi - \beta/2$ by definition.
              On the other hand, $\beta' \leq \beta$.
              Thus we have $\pi - \beta'/2 \geq \pi - \beta/2 \geq \alpha' \geq \beta' \Leftrightarrow \beta' \leq 2\pi/3$

              From \cref{lem:linear} in this case we have that $\Tlinear(\pi - \beta/2, \min\{\beta, \pi - \beta/2\})$ is at most $\pi-\beta/2-\beta'+2\chord(\beta') \leq \pi - 3/2 \beta'+2\chord(\beta')$.
              The derivative of $f(x)=y-3/2x+4sin(x/2)$ is $-3/2+2\cos(x/2)$ is equal to zero for $x=x_0=\arccos(3/4) \equiv 1.44$.
              $f$ is increasing up to $x=x_0$ and then is decreasing.
              Thus $1 + \pi - 3/2 \beta'+2\chord(\beta') \leq 1+ \pi-3/2 x_0 +2 \chord(x_0) < 4.62$.

        \item For robots $\{r_2, \ldots, r_{j-1}\}$ lying on $(\pos_1, \textsc{Q})$.
              Recall that $\alpha_q = 1+2\sqrt{2} + 2 \pi/3-2\sqrt{3}$ is the angular distance from $\pos_1$ to $\textsc{Q}$, and that $r_2$ is awakened at time $1 + |\pos_1\pos2|$.
              We get that the \completionTime{} is $1 + |\pos_1\pos_2| + \Tlinear(\alpha_q - |\pos_1\pos_2|, \beta)$.

              From \cref{lem:linear}, $\Tlinear(\alpha_q - |\pos_1\pos_2|, \beta) \leq \alpha_q - |\pos_1\pos_2| -2\pi/3+2\sqrt{3}$.
              Hence the total \completionTime{} is at most $1  + |\pos_1\pos_2|+ 1+2\sqrt{2} + 2 \pi/3-2\sqrt{3} - |\pos_1\pos_2| -2\pi/3+2\sqrt{3} = 2 + 2\sqrt{2}$.

        \item For robots $\{r_j, \ldots, r_k \}$ lying on the arc $(\textsc{Q} \pos_k)$.
              We upper bound the wake up time of $r_j$ by approximating the movement of $r_0$ as if it does a detour before waking up $r_j$: we consider the geometric path $\pos_0, \pos_1, \pos_2, \textsc{Q}, \pos_j$, that is longer than actual $r_0$'s path.
              Then we upper bound the makespan of $(\pos_j,\pos_k)$ by considering the makespan of the arc $(\textsc{Q} \pos_k)$ as if $r_0$ acts as the first awaked robot, located on $\textsc{Q}$.

              By \cref{lem:sum of chords}, the longest path $r_0$ could possibly make on from $\pos_1$ to $\textsc{Q}$ with $1$ hop consists of two chords of same length and has total length $2 \chord(\alpha_q/2) < 2.31$.

              The length $\alpha''$ of the arc $(\textsc{Q} \pos_k)$ is at most $\pi-\alpha_q=\pi/3-1-2\sqrt{2}+2\sqrt{3} <0.69$.
              Hence the largest empty sector within $(\textsc{Q}, r_k)$ has angle $\beta '' \leq \alpha ''$ which is less than $2\pi / 3$.
              By \cref{lem:linear} the \completionTime{} is bounded by $\alpha '' - \beta '' + 2 \chord(\beta'')$ which is maximal for $\beta '' = \alpha ''$ and $\alpha '' < 0.69$.
              Thus the \completionTime{} of robots $r_j$ to $r_k$ is less than $|\pos_0\pos_1|+|\pos_1 \pos_2|+|\pos_2\textsc{Q}|+2 \chord(0.69) < 4.66 < 2+2\sqrt{2}$.
    \end{enumerate}
\end{proof}

\begin{lemma}\label{lem:sum of chords}
    Let $\theta_1,\dots,\theta_k$ be the non-negative angles of $k$ chords with $\sum_{i=1}^k \theta_k \in [0, 2\pi]$. The sum of the length of these $k$ chords is maximized whenever $\theta_1 = \cdots = \theta_k$.
\end{lemma}


\begin{proof}
    
Let $C$ be the maximum sum of the length of $k$ chords of non-negative angles summing up to some $\theta \in [0,2\pi]$. We want to show that $C = k\cdot \chord(\theta/k)$, i.e., $C$ can be achieved when the $k$ chords have length $\chord(\theta/k)$ each, or equivalently whenever $\theta_i = \theta/k$ for all $i$.

On one side, if all $\theta_i$ are equals, then $\sum_i \chord(\theta_i) = \sum_i \chord(\theta/k) = k \cdot \chord(\theta/k)$, and thus $C \ge k \cdot \chord(\theta/k)$.

On the other side, $\theta \in [0, 2\pi]$ and each $\theta_i \ge 0$, thus $\theta_i \in [0, 2\pi]$. Moreover, the function $x\mapsto \chord(x) = 2 \sin(x/2)$ is concave for $x \in [0, 2\pi]$, since $\sin(u)$ is concave for $u\in [0,\pi]$ as its second derivative $-\sin(u)$ is non-positive on $[0,\pi]$. Thus, by Jensen's inequality on concave functions we have, for all $\theta_i\in [0,2\pi]$ summing up to $\theta \in [0,2\pi]$:
$$
\chord\pare{\frac{1}{k} \sum_{i = 1}^k \theta_i} ~\ge~ \frac{1}{k} \sum_{i=1}^k \chord(\theta_i)  \quad\Longleftrightarrow \quad \sum_{i=1}^k \chord(\theta_i) ~\leq~ k\cdot \chord(\theta/k)
$$
In particular, the last inequality holds for $\theta_i$'s maximizing $\sum_i \chord(\theta_i)$, i.e., when $\sum_i \chord(\theta_i) = C$. Thus $C \le k \cdot \chord(\theta/k)$.

Therefore $C = k \cdot \chord(\theta/k)$, that completes the proof.
\end{proof}

\upperboundSmall*

\begin{proof}
    For $n=3$ the proof is based on the following algorithm:
    Let $\pos_i, \pos_j \in \P\setminus\{\pos_0\}$ be a pair of positions such that the measure $\alpha$ of the angle $\angle \pos_i \pos_0 \pos_j$ is the smallest.
    This measure is at most $ 2 \pi/3$.
    W.l.o.g, we assume that $\pos_i = \pos_1$ and $\pos_j = \pos_2$.
    Robot $r_0$ starts by waking up $r_1$ on $\pos_1$, then wakes up $r_3$ and returns to $\pos_0$.
    The length of this trajectory is below $4$.
    When $r_1$ wakes up (at time at most $1$), it goes to $\pos_2$ and returns.
    Its \completionTime{} is given by the sum of its wake-up time and the length of the path  $\pos_1,\pos_2,\pos_1$.
    It is upper bounded by $1+2 \chord(\alpha) \leq 1+2\chord(2 \pi/3)=1+2\sqrt{3}$.
    This conclude the proof for $n=3$ since other robots $r_2$ and $r_3$ stays still.

    For $n=4$ if points of $\P$ are in convex positions, the proof is a direct application of \cref{th:convex_UB}.
    Hence we only have to prove it for the non-convex distribution of sleeping robots.
    In fact we use Algorithm \convex and provide the analysis for non-convex distribution when $n=4$.
    With four positions, the empty sector of widest angle has a measure $\beta \geq 2\pi/4$, where the two robots positions on its boundary are denoted $\pos_k$ and $\pos_{k+1}$.
    As the the angular distance between $\pos_{k+1}$ and $\pos_1$ is less than $\pi-\beta/2$ we get that it is at most $3 \pi /4$.
    Let us compute the \completionTime{} of $r_0$:
    $r_0$ first goes to $\pos_1$, the closest robot to the bisector of the opposite sector of the empty sector.
    It will continue toward $\pos_k$ with a potential detour through a remaining sleeping robot's positions $\pos_i$.

    We have two subcases:
    \begin{enumerate}
        \item If $\pos_i$ is within the triangle $(\pos_0,\pos_1,\pos_{k})$, it is $r_0$ that wakes up $r_i$.
              $r_0$'s path is thus $\pos_0, \pos_1, \pos_i, \pos_k, \pos_0$.
        \item If $\pos_i$ is within the triangle $(\pos_0,\pos_1,\pos_{k+1})$, then $r_i$'s awakening is attributed to $r_1$.
              Once $r_1$ wakes up $r_i$, it returns to $\pos_1$ while $r_i$ goes to $\pos_{k+1}$ and returns.
    \end{enumerate}

    We now provide the analysis of the strategy.
    In order to bound each paths' length, we upper bound the distance between the subset of positions that are in a convex configuration by assuming these positions lie on the border of the unit disk.
    We denote by $\pos_1'$ (resp. $ \pos_{k}'$ and $\pos_{k+1}'$) the projections of $\pos_1$ (resp. $\pos_{k}$ and $\pos_{k+1}$) on the border of the unit disk centered in $\pos_0$.

    \emph{Distances within a triangle.} We recall the general observations that given a triangle $ABC$ and a point $M$ contained in $ABC$:
    The distance from $A$ to $M$ is at most $\max\{|AB|, |AC|\}$.
    Furthermore, a path $AM, MB$ has length at most $|AC|+|CB|$.
    Given that $A$ is any point of $ABC$, similar inequalities holds when considering distance between $M$ and $B$ or $C$.

    In the first case, we need to upper bound $r_0$'s path to take the detour into account (note that for $r_1$ the analysis of the convex case holds since the strategy is the same).
    We have
    \begin{align*}
        |\pos_0\pos_1| + |\pos_1\pos_i| + |\pos_i\pos_k| + |\pos_k\pos_0| & \leq |\pos_0\pos_1'| + |\pos_1'\pos_i| + |\pos_i\pos_k'| + |\pos_k'\pos_0| \\
                                                                          & \leq 1 + |\pos_1'\pos_i| + |\pos_i\pos_k'| + 1.
        \intertext{Since $\pos_i$ is located within $(\pos_0,\pos_1,\pos_{k})$, it is also contained in $(\pos_0,\pos_1',\pos_{k}')$, thus:}
        |\pos_0\pos_1| + |\pos_1\pos_i| + |\pos_i\pos_k| + |\pos_k\pos_0| & \leq 2 + |\pos_1'\pos_0| + |\pos_0\pos_k'| \leq 4.
    \end{align*}
    This concludes the analysis of $r_0$'s path.

    In the second case, we upper bound the lengths of $r_1$'s and $r_i$'s paths (for $r_0$'s path, the analysis of the convex case holds since the strategy is the same).
    First, we focus on $r_1$'s path $\pos_1, \pos_i,\pos_1$.
    Since its wake up time is $|\pos_0\pos_1|$, we have that $r_1$ ends it journey at time $|\pos_0\pos_1| + 2\cdot|\pos_1\pos_i|$.
    We have that $|\pos_0\pos_1| \leq 1$ and for $|\pos_1\pos_i|$, given that $\pos_i$ is within the triangle $(\pos_0,\pos_1,\pos_{k+1})$, we use the same argument and for the first case's analysis:
    \begin{align*}
        |\pos_1\pos_i| & \leq \max\{|\pos_1'\pos_{k+1}'|, |\pos_1'\pos_0|\}      \\
                       & \leq \max\{\chord(\angle\pos_{k+1}'\pos_0\pos_1'), 1\}.
        \intertext{By construction, $\angle\pos_{k+1}'\pos_0\pos_1' \leq \pi - \beta/2 \leq 3\pi/4$}
        |\pos_1\pos_i| & < 1.85.                                                 \\
        \intertext{We can conclude on $r_1$'s \completionTime:}
        |\pos_0\pos_1| + 2\cdot|\pos_1\pos_i| \leq 1 + 2\cdot 1.85 \leq 4.7.
    \end{align*}

    The last part of the analysis focus on $r_i$'s \completionTime.
    It is at most $t'=|\pos_0\pos_1'|+|\pos_1'\pos_i|+|\pos_i\pos_{k+1}'|+|\pos_{k+1}'\pos_i|$.
    Let $\textsc{Q}$ be the intersection between segment $[\pos_0 \pos_1']$ and the line $(\pos_i\pos_{k+1}')$.
    Clearly, $|\textsc{Q}\pos_{k+1}'| \leq |\pos_i\pos_{k+1}'|$.
    In order to upper bound $t'$, we upper bound $r_i$'s wake up time as if $r_1$ made a detour through $\textsc{Q}$, and we upper bound $r_i$'s path length by analysing the path $\textsc{Q}\pos_{k+1}\textsc{Q}$.

    Since $|Q\pos_i| + |\pos_i\pos_{k+1}'| = |\textsc{Q}\pos_{k+1}'|$ and $|\pos_{k+1}'\pos_i|\leq |\pos_{k+1}'\textsc{Q}|$ we get that $t' \leq |\pos_0\pos_1'|+|\pos_1'\textsc{Q}| + 2\cdot|\textsc{Q}\pos_{k+1}'|$.
    Using the law of cosines within triangle $(P_0QP'_{k+1})$, $|\textsc{Q}\pos_{k+1}'|$'s length can be expressed as a function of $x=|\pos_1'\textsc{Q}|$.
    \begin{align*}
        (f(x))^2 & = \underbrace{(1-x)^2}_{|\pos_0\textsc{Q}|} + \underbrace{1^2}_{|\pos_0\pos_{k+1}|} - 2(1-x)\cos \alpha\quad\text{where $\alpha = \angle \textsc{Q}\pos_0\pos_{k+1}$} \\
                 & \leq 2 + x^2 - 2x - 2(1 - x)\cos(3\pi / 4)\quad \text{ which is decreasing on $[0,1]$}                                                                                \\
        (f(x))^2 & \leq (f(0))^2 = 2 - 2 \cos(3\pi / 4) \quad\text{ since $x = |\pos_1'\textsc{Q}| \leq |\pos_0\pos_1'| = 1$}                                                            \\
        f(x)     & < \sqrt{3.42} < 1.85.
    \end{align*}

    Thus $r_i$'s \completionTime{} $t'$ is at most $1 + 2 \cdot f(0) < 4.7$.

    It turns out that the non-convex distribution of $4$ sleeping robots leads to a makespan strictly smaller than $2+2\sqrt{2}$. 
\end{proof}

\subsection{Asymptotic Upper Bound}\label{appendix:AUB}

\upperboundnlarge*

\cyril{In the proof just above, I see it is used $r_i$'s "completion time" (makespan?), and here bellow the "return time" $\returntime(r_i)$. It would be nice to uniformize at least this two proofs.}\nicolasB{Both notations were used in the paper. I made a command \\completionTime to uniformize.}

Before turning to the proof of the theorem, we record a corollary of~\cite[Proposition~8]{BCGH24} that will be used in Case~1 below to bound the wake-up time of the sub-cone~$C_{b_i}$. The setting is the following: the initially awake robot is no longer located at the apex of a cone, but at the inner corner of a crown.

\begin{corollary}[adapted from {\cite[Proposition~8]{BCGH24}}]\label{cor:splitcone-inner-corner}
    Let $\alpha \in (0, 2\pi]$ and $w \in [0,1]$, and let $C$ be a crown of angle $\alpha$ and width $w$ in the unit disk centered at $\pos_0$. Let $P$ be a set of $n$ points contained in $C$, and let $p_0 \in P$ be located at an inner corner of $C$ (i.e.\ $|\pos_0 p_0| = 1-w$). One can construct in time $O(n \log n)$ a wake-up tree for $P$ rooted at $p_0$ whose makespan is at most $w + \varphi\,\alpha$.
\end{corollary}

\begin{proof}
    Let $C^\star$ be the cone of apex $\pos_0$ and arc-length $\alpha$ at unit radius that contains $C$; the crown $C$ is precisely the annular part of $C^\star$ at distances in $[1-w, 1]$ from $\pos_0$. Apply \splitconestrategy to the point set $P$ inside $C^\star$ with initial awake robot located at the apex $\pos_0$. By~\cite[Proposition~8]{BCGH24}, this constructs in time $O(n\log n)$ a wake-up tree $T'$ rooted at $\pos_0$ for $P$, with makespan at most $1 + \varphi\,\alpha$.

    The first edge of $T'$ goes from $\pos_0$ to a nearest neighbour of $\pos_0$ in $P$. Since every point of $P$ lies at distance at least $1-w$ from $\pos_0$ and $p_0$ lies at distance exactly $1-w$, the nearest neighbours of $\pos_0$ in $P$ are precisely the points of $P$ at distance $1-w$, all of which are inner corners of~$C$ and interchangeable in the analysis. Choosing the tie-breaking rule of \splitconestrategy that picks $p_0$ first (this does not affect the makespan bound of~\cite[Proposition~8]{BCGH24}), the first edge of $T'$ is $(\pos_0, p_0)$, of length $1-w$. Let $T$ be the subtree of $T'$ rooted at $p_0$, obtained by removing $\pos_0$ and its outgoing edge. Then $T$ is a wake-up tree for $P$ rooted at $p_0$, and
    \[
        \mathrm{makespan}(T) = \mathrm{makespan}(T') - |\pos_0 p_0| \le (1 + \varphi\,\alpha) - (1-w) = w + \varphi\,\alpha.
        \qedhere
    \]
\end{proof}

\begin{proof}
    Recall that the initial robot is at the origin and all robots are within the unit disk.
    Let $\alpha_b \in [0.94, 0.95]$ be a parameter whose value will be fixed
    \emph{a posteriori} so as to balance the upper bounds derived below for the
    different types of robots (see equations \eqref{eq:returntime-narrow},
    \eqref{eq:returntime-cr} and \eqref{eq:returntime-rk}). We set
    $\rho_b = \chord(\pi - \alpha_b/2) = 2\cos(\alpha_b/4)$ and
    $\rho_c = \rho_b - 1 \in [0,1]$. The optimum is attained for
    $\alpha_b \approx 0.9439742895$, which gives $\rho_b \approx 1.9445650300$.

    The algorithm is as follows:

    Phase 1:
    Divide the disk into $\sqrt{n}$ cones centered at $\pos_0$ of equal angle
    $\alpha = 2\pi / \sqrt{n}$. Let $C_0$ be the cone containing the most robots.
    W.l.o.g., assume that the bisector of $C_0$ coincides with the negative
    $y$-axis. Wake up all robots in $C_0$ using the \splitconestrategy
    algorithm~\cite[Proposition~8]{BCGH24}.
    Since $C_0$ contains the largest number of robots among $\sqrt{n}$ cones
    partitioning $n$ sleeping robots, it contains at least $\sqrt{n}$ of them,
    which is enough for the $\sqrt{n}-1$ wake-ups to come.

    Phase 2 (see \cref{fig:strat_asymp}(a)): Let $C_b$ be the cone opposite
    $C_0$ with angle $\alpha_b$. Let $Q$ be the intersection of the bisector of
    $C_0$ with the unit circle centered at $\pos_0$.
    Starting from $Q$, define a fan of $\sqrt{n}-1$ \emph{narrow} cones
    \emph{with apex at $Q$}, of equal angle $\beta_n = \pi/(\sqrt{n}-1)$,
    whose union covers the unit disk minus $C_b$.
    Indeed, since $Q$ lies on the unit circle, the unit disk spans an angle
        of exactly $\pi$ from $Q$, hence the total angle $(\sqrt{n}-1)\beta_n = \pi$.
    The radius of each narrow cone is at most the maximum distance from $Q$
    to a point of the unit disk not in $C_b$, which is reached at the edges of
    $C_b$ and equals $\chord(\pi - \alpha_b/2) = \rho_b$.
    For each narrow cone, one of the robots awakened (other than $r_0$) in
    Phase~1 wakes up all sleeping robots in this cone using the
    \splitconestrategy algorithm before returning to its initial position.

    \begin{figure}[htbp!]
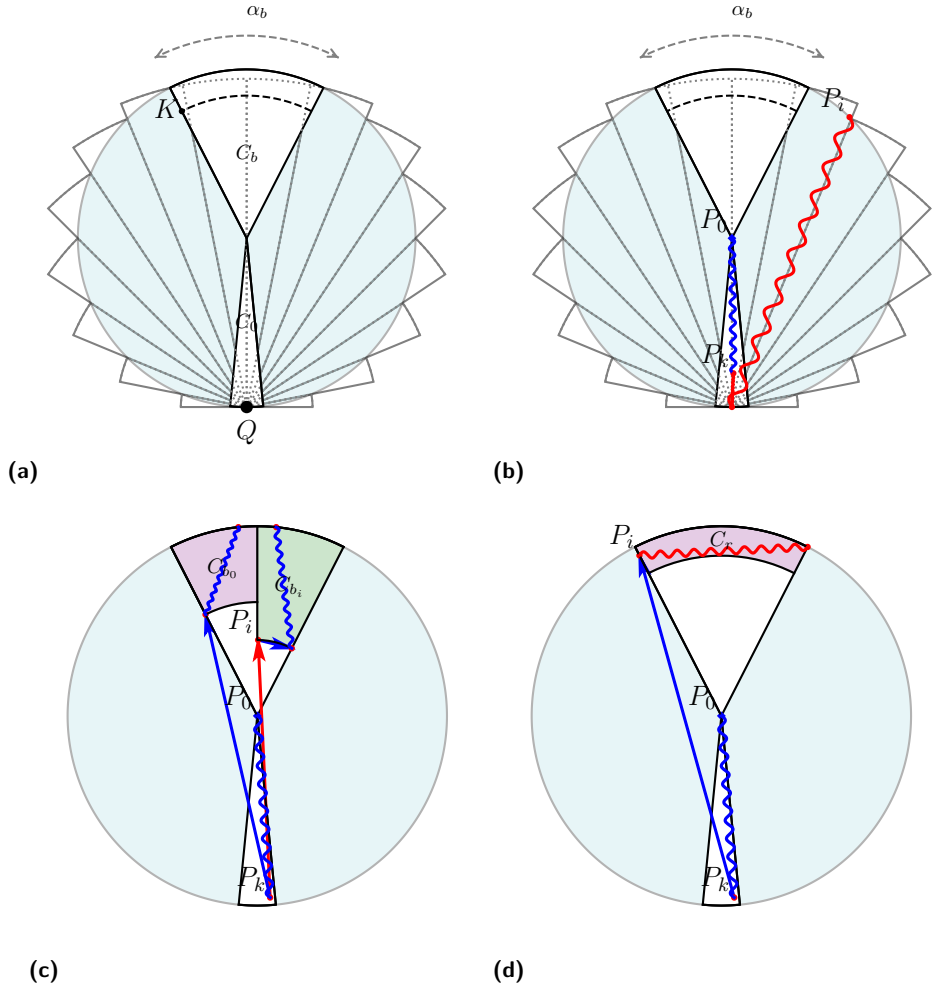

        \centering
        \def\SPC{\vspace{-3ex}}
        \begin{subfigure}[b]{0.45\textwidth}
            \incfig{\FigLargeNEmpty}{1}
            \SPC\caption{}%
            \label{fig:LargeNEmpty}
        \end{subfigure}
        \begin{subfigure}[b]{0.45\textwidth}
            \incfig{\FigLargeN}{1}
            \SPC\caption{}%
            \label{fig:LargeN}
        \end{subfigure}

        \begin{subfigure}[b]{0.43\textwidth}
            \incfig{\FigLargeNCaseOne}{1}
            \SPC\caption{}%
            \label{fig:largeNCaseOne}
        \end{subfigure}
        \begin{subfigure}[b]{0.43\textwidth}
            \incfig{\FigLargeNCaseTwo}{1}
            \SPC\caption{}%
            \label{fig:largeNCaseTwo}
        \end{subfigure}
        \caption{(a) Decomposition of the disk. (b) Phase~2 for narrow cones. (c) Case~1 for the big cone. (d) Case~2 for the big cone.}%
        \label{fig:strat_asymp}
    \end{figure}

    For cone $C_b$, we apply a specific wake-up strategy involving $r_0$, distinguishing two cases:
    \begin{itemize}
        \item Case~1 (see \cref{fig:strat_asymp}(c)): Cone $C_b$ contains at
              least one initial position $\pos_i$ at distance at most $\rho_c$
              from the initial robot $r_0$. W.l.o.g.\ assume that $\pos_i$ is
              on the right half of $C_b$. Then $r_k$, one of the robots
              awakened in $C_0$, wakes up $r_i$ before returning to its initial
              position. Then, $r_i$ and $r_0$ split $C_b$ into two
              sub-cones $C_{b_0}$ and $C_{b_i}$, each of opening $\alpha_b/2$
              (with apex $\pos_0$ for $C_{b_0}$ and apex $\pos_i$ for
              $C_{b_i}$), and each wakes up the robots of its sub-cone using
              \splitconestrategy before returning to its initial position.

        \item Case~2 (see \cref{fig:strat_asymp}(d)): Cone $C_b$ does not
              contain any initial position at distance at most $\rho_c$ from
              $\pos_0$. Then all initial positions in $C_b$ are located within
              a crown $C_r$ of angle $\alpha_b$ and width $1-\rho_c$. After
              Phase~1, $r_0$ wakes up the robot $r_i$ that is the closest to
              the left edge of $C_b$. Once this is done, $r_0$ returns directly
              to $\pos_0$ while $r_i$ wakes up all robots in $C_r$ using the
              \crown algorithm. When all robots are awakened, they return to
              their respective initial positions.
    \end{itemize}

    This algorithm has complexity $O(n \log n)$ since each phase uses
    \splitconestrategy and \crown, both of complexity $O(n \log n)$.

    Let us now analyze the makespan.
    Recall that $\mathrm{crown}(\alpha, w)$ denotes a crown of opening $\alpha$
    and width $w$, $\diam(R)$ denotes the diameter of a connected region
    $R \subset \mathbb{R}^2$, and the diagonal of $\mathrm{crown}(\alpha, w)$ is
    $\gchord(\alpha, w) = \sqrt{1 + (1-w)^2 - 2(1-w)\cos\alpha}$. We have
    $\diam(\mathrm{crown}(\alpha, w)) = \max\{\chord(\alpha), \gchord(\alpha,w)\}$.
    The wake-up time of \splitconestrategy on a cone of angle $\alpha$ and
    radius $\rho$, with an awake robot at its apex, is
    $\cone(\alpha, \rho) = \rho + \varphi\alpha$.
    The wake-up time of \crown on a crown of angle $\alpha$ and width $w$,
    with an awake robot on a straight side, is
    $\crown(\alpha, w) = \alpha + (1+\varphi) w$.

    The duration of Phase~1 is $\cone(\beta_n, 1) = 1 + O(\beta_n)$.
    Within this time every awake robot can reach the border of $C_0$; we may
    therefore assume that at the end of Phase~1 all awake robots meet at $Q$
    at time $1 + O(\beta_n)$.

    We analyze the return time of each robot depending on its initial position.

    \begin{itemize}

        \item For a robot $r_l$ in a narrow cone (including $C_0$) that is not involved in the awakening process of cone $C_b$, the time needed to reach $\last_{r_l}$ is given by the wake-up time of a narrow cone: $\cone(\beta_n, \rho_b)$.
              The time needed to return to its initial position is bounded by the diameter of the union $U$ of this single narrow cone and the cone $C_0$, which we bound by $\rho_b + O(\beta_n)$ as follows.
              Let $M$ be the second intersection of the narrow cone's bisector with the unit circle, and consider the triangle $T = \pos_0 Q M$.
              First, since $\pos_0$ is the center of the unit disk and $Q, M$ lie on the unit circle, $|\pos_0 Q| = |\pos_0 M| = 1$; moreover $M$ is a point of the unit disk outside $C_b$, so $|QM| \le \rho_b$. As $\rho_b > 1$, the diameter of $T$ is $\max\set{1, 1, |QM|} = |QM| \le \rho_b$.
              Second, every point of $U$ lies within distance $O(\beta_n)$ of $T$: the cone $C_0$ has apex $\pos_0$, half-angle $\alpha/2 = O(\beta_n)$ and radius at most $1$, hence lies within distance $O(\beta_n)$ of the side $[\pos_0, Q]$; the narrow cone has apex $Q$, half-angle $\beta_n/2$ and radius at most $\rho_b$, hence lies within distance $\rho_b \sin(\beta_n/2) = O(\beta_n)$ of the side $[Q, M]$.
              Combining the two, $\diam(U) \le \diam(T) + O(\beta_n) \le \rho_b + O(\beta_n)$, so that $|\last_{r_l} \pos_{r_l}| \le \rho_b + O(\beta_n)$.
              The return time of $r_l$ is bounded by
              $\returntime(r_l) \leq 1 + O(\beta_n) + \cone(\beta_n, \rho_b) + \rho_b = 1 + 2 \rho_b + O(\beta_n)$.

              Hence, we get:
              \begin{equation}\label{eq:returntime-narrow}
                  \returntime(r_l) \le 1 + 4\cos(\alpha_b/4) + O(\beta_n).
              \end{equation}

        \item For a robot $r_l$ in cone $C_{b_i}$ (including $r_i$) in Case~1.
              Let $y = |\pos_0 \pos_i| \le \rho_c$. Since $\pos_i$ is the robot of $C_b$ closest to the origin, every robot of $C_b$ is at distance $\ge y$ from $\pos_0$.
              After Phase~1, $r_k$ walks from $Q$ to $\pos_i$, covering at most $|Q\pos_i| \le 1 + y$, hence $r_i$ is awakened by time $1 + O(\beta_n) + (1+y) = 2 + y + O(\beta_n)$.

              The set of robots in $C_{b_i}$ is contained in the half-crown $H$ of angle $\alpha_b/2$ and width $1-y$ (the right half of $C_b$ restricted to radii in $[y, 1]$), and $\pos_i$ lies at the inner corner of $H$. By~\cref{cor:splitcone-inner-corner}, we can construct a wake-up tree for these robots, rooted at $\pos_i$, of makespan at most $(1-y) + \varphi\,\alpha_b/2$.

              The return arc $|\last_{r_l}\pos_{r_l}|$ is bounded by the diameter of $H$:
              \[
                  |\last_{r_l}\pos_{r_l}| \le \diam(H) = \diam(\mathrm{crown}(\alpha_b/2, 1-y)) = \max\set{\gchord(\alpha_b/2, 1-y), \chord(\alpha_b/2)}.
              \]

              Summing arrival, wake-up and return, the $y$ contributions cancel:
              \begin{align*}
                  \returntime(r_l) &\le (2+y) + \bigl((1-y) + \varphi\,\alpha_b/2\bigr) + \max\set{\gchord(\alpha_b/2, 1-y), \chord(\alpha_b/2)} + O(\beta_n)\\
                  &= 3 + \varphi\,\alpha_b/2 + \max\set{\gchord(\alpha_b/2, 1-y), \chord(\alpha_b/2)} + O(\beta_n).
              \end{align*}
              For $\alpha_b \in [0.94, 0.95]$ and $y \in [0,\rho_c]$, the bracket is bounded by $1$: indeed $\gchord(\alpha_b/2, 1-y) \le \gchord(\alpha_b/2, 1) = 1$, and $\chord(\alpha_b/2) = 2\sin(\alpha_b/4) \le 2\sin(0.95/4) < 0.47$. The maximum is attained at $y = 0$, where $\gchord(\alpha_b/2, 1) = 1$. Therefore
              \begin{equation}\label{eq:returntime-cbi}
                  \returntime(r_l) \le 4 + \varphi\,\alpha_b/2 + O(\beta_n).
              \end{equation}
              Numerically this is below $4.77$, hence strictly below the other three bounds.

        \item For a robot $r_l$ in cone $C_{b_0}$ (including $r_0$) in Case~1: the wake-up of cone $C_{b_0}$ starts at $1 + O(\beta_n) + 1$ and lasts $\cone(\alpha_b /2,1)$. The return arcs are bounded by the diameter of cone $C_{b_0}$, which is at most 1 for values of $\alpha_b \in [0,1]$.
              Hence $\returntime(r_l)$ is bounded by expression~\eqref{eq:returntime-cbi} as well.

        \item For the robot $r_0$ in Case~2: after Phase~1, $r_0$ goes to $\pos_i$ and then goes back to $\pos_0$. Hence $\returntime(r_0) \leq 1+ O(\beta_n) + 2 + 1 = 4 + O(\beta_n)$.

        \item For a robot $r_l$ in the crown $C_r$ (including $r_i$) in Case~2.
              Robot $r_0$ goes from $Q$ to the left inner corner $K$ of $C_r$ and then starts the wake-up of the crown.
              Since $K$ is at distance $\rho_c$ from $\pos_0$ on a ray making an angle $\pi - \alpha_b/2$ with the ray $\pos_0 Q$, the law of cosines gives
              \[
                  |QK| = \sqrt{1 + \rho_c^2 + 2\rho_c \cos(\alpha_b/2)}.
              \]
              Putting everything together, we get
              \begin{equation}\label{eq:returntime-cr}
                  \returntime(r_l) \leq 1 + \sqrt{1 + \rho_c^2 + 2\rho_c \cos(\alpha_b/2)} + \alpha_b + (1+\varphi)(1-\rho_c) + 2 \sin(\alpha_b/2) + O(\beta_n)~.
              \end{equation}

        \item For the robot $r_k$ that wakes $r_i$ in Case~1: after Phase~1, $r_k$ goes from $Q$ to $\pos_i$, then back to $\pos_k$. By definition of $r_i$ in Case~1, $|Q\pos_i| \leq 1 + y \le \rho_b$ and $|\pos_i\pos_k| \leq 1 + y \le \rho_b = 2\cos(\alpha_b/4)$, hence
              \begin{equation}\label{eq:returntime-rk}
                  \returntime(r_k) \leq 1 + O(\beta_n) + |Q\pos_i| + |\pos_i\pos_k| \le 1 + 4\cos(\alpha_b/4) + O(\beta_n) ~.
              \end{equation}
    \end{itemize}
    The chosen values for $\alpha_b$ and $\rho_c$ balance the return times for robots in the narrow cones, the crown $C_r$, and the robot $r_k$. In other words, they are chosen such that the three expressions~\eqref{eq:returntime-narrow}, \eqref{eq:returntime-cr},  and~\eqref{eq:returntime-rk} are equal.

    Using the fact that $\beta_n = \Theta(1/\sqrt{n}\,)$, one can verify that for all robots the return time is strictly below $4.88913006 + O(1/\sqrt{n}\,)$. This completes the proof.
\end{proof}

\subsection{General Upper Bound}\label{appendix:GUB}

This section is dedicated to present the proofs of the Lemmas used to prove \cref{th:generalUB} in \cref{sec:generalUB}.

\generalUB*

Let us first recall some notations introduced in \cref{sec:generalUB}.
\begin{itemize}
    \item $\pos_r$ is the initial position of $r \in \R$;
    \item $\Bad$ is the set of robot $r$ such that the arc $(\last_r, \pos_r)$ is an $\eps$-bad return arc.
    \item Given a robot $r \in \Bad$:
          \begin{itemize}
              \item $\good_r$ is the \emph{closest good ancestor} of $\last_r$;
              \item $\bad_r$ is the child of $\good_r$ leading to $\last_r$ and belonging to $\eps$-bad pairs;
              \item $\wutree_{\bad_r}$ subtree of $\wutree$ rooted at $\bad_r$.
              \item $\Lset_r$ is the set of the robots in $\wutree_{\bad_r}$
              \item $\best_r$ is the \emph{best ancestor} of $r$, with $\Rbest$ the corresponding robot.
              \item $\worst_{\Rbest}$ is the \emph{worst descendant} of $\best_r$, $\Rworst$ the corresponding robot.
              \item $\Bad_\Rbest$ is the subset of robots in~$\Bad$ who share the same best ancestor~$\Rbest$.
          \end{itemize}
\end{itemize}

\lemmaOnCrown*
\begin{proof}
    Let~$\wutree_{\bad_r}$ be the subtree of~$\wutree$ rooted in~$\bad_r$.

    First, we show that $|\pos_{r}\pos_{x}| > 2-\eps$,
    which implies that~$x$ lies on the $\eps$-crown.
    Recall that by definition, $(\bad_r,r)$ is an $\eps$-bad pair: $\wutime_\wutree(\bad_r)+|\pos_{r}\bad_r|>t + 2 - \eps$.
    In particular, $\wutime_\wutree(\bad_r) > t-\eps$.
    As~$\bad_r$ is an ancestor of~$\pos_x$, $\wutime_\wutree(\bad_r)+|\bad_r \pos_{x}| \leq \wutime_\wutree(\pos_x)$.
    Since $\wutime_\wutree(\pos_x) \leq t$, we have $\wutime_\wutree(\bad_r)+|\bad_r\pos_{x}| \leq t$.
    If~$|\pos_{r}\pos_{x}| \leq 2-\eps$, then by triangle inequality, it holds true that $|\bad_r\pos_{r}| \leq |\bad_r\pos_{x}| + |\pos_{r}\pos_{x}|  \leq 2-\eps + |\bad_r\pos_{x}|$.
    This implies~$\wutime_\wutree(\bad_r)+|\bad_r\pos_{r}| \leq \wutime_\wutree(\bad_r) + |\bad_r\pos_{x}| + |\pos_{r}\pos_{x}| \leq t+2-\eps$, raising a contradiction. 
    Consequently, $|\pos_{r}\pos_{x}| > 2-\eps$, and therefore~$x$ is on the crown.

    We now show that there exists a leaf~$\pos_\ell$ of~$\wutree_{\bad_r}$, such that~$|\pos_{x}\pos_\ell|\leq \eps$.
    If~$\pos_x$ is a leaf, the property clearly holds for~$\pos_\ell = \pos_x$.
    Otherwise, let~$\pos_\ell$ be any leaf of $\wutree_x$. 
    As $\wutime_\wutree(\bad_r) >t-\eps$ and $\wutime_\wutree(\pos_\ell) \leq t$, the wake up time of~$x$ differs by less than~$\eps$ from the wake up time of~$\ell$.
    As~$\pos_\ell$ is a descendant of~$\pos_x$, we conclude that~$|\pos_{x}\pos_\ell| < \eps$.
\end{proof}

\lemmaCrown*
\begin{proof}
    The proof is based on the wake-up time for $\FTP$, noted $\crown(\alpha,w)$, within a crown $C$ of angle $\alpha$ and width $w$, starting from a robot $r_0$ located on a side of the crown given in~\cite{BGHO24}: $\crown(\alpha,w) = \alpha+(1+\varphi)w$.

    We now assume that the initially awake robot can start from any position $\pos_0$ within a crown. We consider the following algorithm: let $\pos_1$ be the position whose angular distance to $\pos_0$, $h$, is the smallest.
    Let $C_0$ be the sub-crown of $C$ containing points at an angular distance smaller to $h$ from $\pos_0$.
    $C$ is partitioned into three crowns $C_1$ (containing $\pos_1$), $C_0$ (containing  $\pos_0$) and $C_2$ the remaining part.
    The angle of $C_0$ is $2h$, and the angles of $C_1$ and $C_2$ are smaller than $\alpha-2h$.
    Let $\pos_2$ be the position within $C_2$ with the smallest angular distance to $\pos_0$, and let $h+h'$ be that angular distance.
    The awakening starts with $r_0$ waking up $r_1$, then $r_2$ and then coming back to $\pos_0$.
    In $C_1$ and $C_2$, we run the algorithm presented in \cite{BGHO24}, starting from $\pos_1$ and $\pos_2$ with $r_1$ and $r_2$, respectively, as source robot.
    This wakes-up both crowns but does not include the returns.
    Yet, these returns are bounded by the diameters of the sub-crowns, being smaller than $\alpha-2h+w$ ($C_1$) and $\alpha-(2h+h')+w$ ($C_2$).

    In our analysis, we consider that a slightly longer path reaches $r_2$, namely path \mbox{$L_2 = \pos_{0},\pos_{1},Q,\pos_{2}$} where
    $Q$ is the geometric point at the same distance from $\pos_0$ than $\pos_1$ but on the boundary of $C_2$. 
    The length of $L_2$ satisfies $|L_2| \leq (w+h) + 2h + (h'+w) \leq 3h+h'+2w$, which induces an upper bound on makespan of $C_2$:
    $M(C_2) \leq |L_2| + \crown(\alpha-(2h+h'),w)+\alpha-(2h+h') + w \leq 3h+h'+2w + 2(\alpha-2h-h') + (1+\varphi)w + \alpha-2h-h'+w \leq 2\alpha - h - 2h' + (4+\varphi)w \leq 2\alpha + (4+\varphi)w$.
    This upper bound also holds for $C_1$ since we have $M(C_1) \leq |\pos_0,\pos_{1}| + \crown(\alpha-2h,w)+\alpha-2h + w \leq w+h + \alpha - 2h + (1+\varphi)w + \alpha - 2h + w \leq 2\alpha + (3+\varphi)w - 3h \leq 2\alpha + (3+\varphi)w$

    Finally, we must compute the \completionTime{} $\returntime(r_0)$ for $r_0$ to return to its initial location, following the path $L_0 =  \pos_{0},\pos_{1},Q,\pos_{2},\pos_{0}$
    $\returntime(r_0) = |L_0| \leq |L_2|+\alpha-h+w \leq 3h+h'+2w + \alpha-h+w \leq \alpha + 2h + h' + 3w$.
    Since $h' \leq \alpha - 2h$ we obtain  $\returntime(r_0) \leq 2\alpha+3w$, and concludes the proof.
    %
    %
\end{proof}

\subparagraph{Trajectories.}%
In \cref{sec:model} we presented as solutions of $\FTRPid$ pairs $(\wutree, \returnarc)$ where $\wutree$ is a wake-up tree of $\P$ and $\returnarc \subseteq \P^2$ is a set of return arcs.
We present below an alternative and quasi-equivalent formalism, which will be helpful in the proof of \cref{lem:best ancestor}.
Given a wake up tree~$\wutree$, we call~$\lambda \colon \wakeuparc \to \{0,1\}$ a~\emph{trajectory labeling} of~$\wutree$, if
(i)~$\lambda(e) = 0$ for the unique outgoing arc~$e$ of the source robot,
(ii)~$\lambda(e_1) \neq \lambda(e_2)$ for each internal node of~$\wutree$ with two outgoing arcs~$e_1$ and~$e_2$, and
(iii)~$\lambda(e_1) = 1$ for each unique outgoing arc of any internal vertex with only one outgoing arc.
Such labeling allows us to specify which robot traverses which arc of the wake up tree.
As a convention, $\lambda(\pos_x,\pos_y) = 0$ means that~$x$ goes to~$y$, whereas~$\lambda(\pos_x,\pos_y) = 1$ means that the robot that woke up~$x$ goes to~$y$ instead.
Note that besides internal vertices with only one outgoing arc in~$\wutree$, each robot follows a path to a leaf vertex.

Without condition (iii), we would have a strict equivalence between the return arc formalism and the trajectory labeling formalism, since
given a solution $(\wutree, \returnarc)$, we can forge $\lambda$ by labeling the arcs such that for any robot $r$, on the path leading from $\pos_r$ to $\last_r$, the first edge is labeled~$0$ and the other one are labeled~$1$.
Conversly, given a trajectory labeling it is easy to find which position a robot ends at, and therefore to define the return arc between that last position and its initial position.
Condition (iii) reduces the number of solution we consider.
Namely, we only consider at first the solutions of $\FTRPid$ in which, apart from self-loops, return arcs always start at a leaf of $\wutree$.
This is necessary to guarantee that the robot initially located on $\last_{\Rbest}$ can actually initiate a wake-up of the $\eps$-crown, since otherwise this robot would already have a different task in the wakeup tree.
Note that it could be that the labelling we eventually produce does not satisfy this third requirement, and this would be consistent with the construction and the proof.

\subparagraph{Trajectories Swap.}%
Let~$\lambda$ be a trajectory labeling of a wake up tree~$\wutree$ and let~$\pos_v$ be an internal vertex of~$\wutree$ with two outgoing arcs~$e_1$ and~$e_2$.
We define the~\emph{trajectory swap of~$\lambda$ at~$\pos_v$} as the trajectory labeling~$\lambda'$ of~$\wutree$ obtained from~$\lambda$ by replacing the label of~$e_1$ and~$e_2$.
Let~$r$ be the robot under~$\lambda$ that wakes up~$v$ and let~$(\last_r,\pos_r)$ and~$(\last_v,\pos_v)$ be the return arcs of~$r$ and~$v$ under~$\lambda$ respectively.
Note that under~$\lambda'$, the return arc of~$r$ is~$(\last_v,r)$ and the return arc of~$v$ is~$(\last_r,v)$.

\subparagraph{Defining $\Rbest$.}
Let us now define $\Rbest$ and its initial position $\best_r$ under the assumption that $\eps < 1$.
Given $r\in\Bad$, the \emph{best ancestor of $r$}, $\best_r$ is the closest ancestor of $\pos_r$ in $\wutree$ such that $(\last_{\Rbest},\best_r)$ is an $\eps$-good return arc and such that $\Rbest$ wakes up $r$ or a robot located at an ancestor of $\pos_r$.
Such an ancestor exists since $r_0 \notin \Bad$ and $\pos_{r_0}$ is an ancestor of all nodes of $\wutree$.
We now justify that $\last_{\Rbest}$ satisfies two conditions.

The first condition is that robot $r'$ located at $\last_{\Rbest}$ does not belong to $\Lset$.
First, note that $\forall r \in \Bad, \Lset_r \cap \Bad = \emptyset$.
Indeed, for a given $r\in \Bad$, $(\bad_r,\pos_r)$ is an $\eps$-bad pair.
This implies that $\wutime_\wutree(\bad_r) > t-\eps$ and, therefore, that $\wutree_{\bad_r}$ has depth at most $\eps$. Since $r$ was the only robot crossing the arc $(\good_r,\bad_r)$, $\wutree_{\bad_r}$ cannot contain any return arc with length greater that $\eps < 2-\eps$.
From this, since $\best_r$ is an ancestor of $r\in\Bad$ it follows that $\best_r \notin \Lset$.
Suppose now that there exists $r''$ such that $r' \in \Lset_{r''}$.
Subtree $\wutree_{\bad_{r''}}$ can not contain both endpoints of the outgoing edge $(\last_{\Rbest},\best_r)$ since it would imply $\Rbest \in \Lset_{r''}$.
So there has to be two distinct outgoing edges from $\wutree_{\bad_{r''}}$: one from $\last_{\Rbest}$ to $\best_r$ and one from $\last_{r''}$ to $\pos_{r''}$.
Lastly we rule out the possibility that $(\last_{\Rbest}, \best_r)$ and $(\last_{r''}, \pos_{r''})$ are in fact an unique arc: this would be the case if $r''$ and $\Rbest$ were the same robot.
However, by definition $\Rbest \not\in \Bad$, and by hypothesis $r'' \in \Bad$, so $\Rbest$ and $r''$ cannot be the same robot.

The second condition is that $\last_{\Rbest}$ is close to robots of $\Lset_r$.
We exhibit one robot, far enough from both $\last_{\Rbest}$ and robots of $\Lset_r$, to guarantee that they lie in the same area of the unit disk.
For every~$\best_r$ being the best ancestor of at least one robot $r$, let~$\Rworst$ located at $\worst_\Rbest$ be the first robot on the path from~$\best_r$ to~$\last_{\Rbest}$ that belongs to $\Bad$.
Such a descendant exists, because either $\Rbest$ wakes up $r$, and therefore $r$ is candidate for being $\Rworst$, either $\worst_\Rbest$ is on the path from $\best_r$ to $\last_\Rbest$ and is an ancestor of $\pos_r$.
Robot~$\Rworst$ necessarily belongs to $\Bad$ by definition of $\best_r$.
We call~$\worst_\Rbest$ the~\emph{worst descendant of~$\best_r$}.
Lemmas~\ref{lem:best ancestor} and~\ref{small cone contains Lr bis} guarantee that $\worst_\Rbest$ is a valid intermediate to guarantee the second condition.

\begin{restatable}{lemma}{lemBestAncestor}%
    \label{lem:best ancestor}
    Let~$\best_r$ be the best ancestor of at least one robot $r$. 
    Then, $|\worst_\Rbest \last_{\Rbest}| > 2-\eps$ and for each leaf~$\pos_\ell$ of~$\wutree_{B_r}$ we have~$|\worst_\Rbest \pos_{\ell}| > 2-\eps$.
\end{restatable}

\begin{figure}[htbp!]
    \centering
    \includegraphics[trim={50 70 50 75}, clip, width=0.7\textwidth]{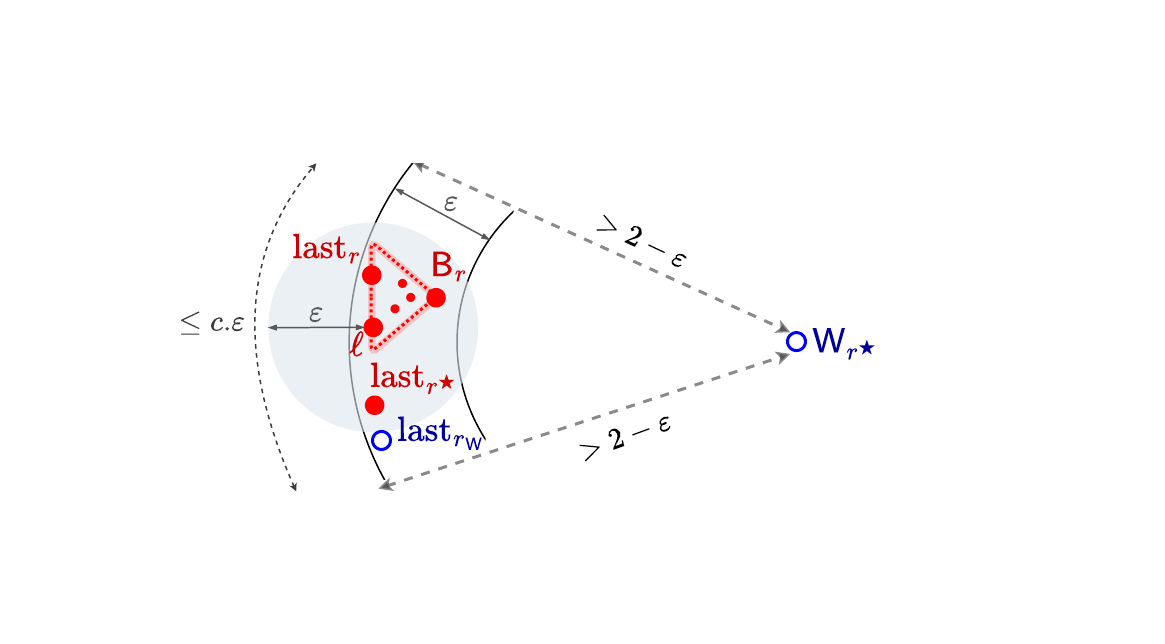}
    \caption{}%
    \label{fig:GUBclose}
\end{figure}

\begin{proof}
    \cref{fig:GUBclose} represents the situation.
    Recall that~$(\last_{\Rworst}, \worst_{\Rbest})$ is an $\eps$-bad return arc and that~$(\last_\Rbest,\best)$ is an $\eps$-good return arc.

    Let us first prove that~$|\worst_\Rbest \last_\Rbest| > 2-\eps$.
    To this end, consider a trajectory swap performed on~$\worst_{\Rbest}$ and~$\best_r$.
    After performing this swap, we have replaced the return arcs~$(\last_{\Rworst},\worst_{\Rbest})$ and~$(\last_\Rbest,\best_r)$ by the return arcs~$(\last_{\Rworst},\best_r)$ and~$(\last_\Rbest,\worst_\Rbest)$.
    We will show that~$(\last_\Rbest,\worst_\Rbest)$ is an $\eps$-bad return arc.
    Note that not both new return arcs can be $\eps$-good, as this would contradict that the previous trajectory labeling achieves the fewest number of $\eps$-bad return arcs.
    So, at least one of the two return arcs is $\eps$-bad.
    If~$(\last_{\Rworst},\best_r)$ is an $\eps$-good return arc, this already implies that~$(\last_{\Rbest},\worst_\Rbest)$ is an $\eps$-bad return arc.
    So assume that~$(\last_{\Rworst},\best_r)$ is an $\eps$-bad return arc.
    If~$(\last_{\Rbest},\worst_\Rbest)$ was an $\eps$-good return arc, then this would contradict the assumption that the previous trajectory labeling achieves the largest $\eps$-cost.
    Hence, $(\last_{\Rbest},\worst_\Rbest)$ is an $\eps$-bad return arc too.
    In both cases, $(\last_{\Rbest},\worst_{\Rbest})$ is an $\eps$-bad return arc, which implies that~$|\worst_{\Rbest} \last_\Rbest| > 2-\eps$.

    Secondly, we denote by $\Bad_\Rbest$ the set of all robot with an $\eps$-bad return arc who share the same best ancestor $\Rbest$.
    Let~$r\in \Bad_\Rbest$ and let~$\pos_\ell$ be a leaf of~$\wutree_{\bad_r}$.
    By definition, $\bad_r$ is woken up by~$r$.
    Moreover, the robots in~$\wutree_{\bad_r}$ have pairwise distance at most~$2\eps$ from each other (see the proof of~\Cref{lem:rb on crown}).
    Since~$\eps < 0.5$, this implies that all return arcs that start from a robot of~$\Lset_r\setminus \{\last_r\}$ are $\eps$-good.
    Now, there is a sequence of trajectory swaps on vertices inside~$\Lset_r$, such that the resulting trajectory labeling contains the $\eps$-bad return arc~$(\pos_\ell,\pos_r)$.
    All the other return arcs that differ between the old and the new trajectory labeling are still $\eps$-good.
    Hence, the new trajectory labeling has the same number of $\eps$-bad return arc and the same $\eps$-cost.

    We now show that~$(\pos_\ell,\worst_\Rbest)$ is an $\eps$-bad pair.
    This then implies that~$|\worst_\Rbest \pos_\ell| > 2-\eps$.
    To this end, first note that~$\pos_r$ and~$\pos_\ell$ are in the subtree of~$\wutree$ rooted in~$\worst_{\Rbest}$.
    Indeed, if the closest common ancestor~$\pos_q$ of~$\pos_r$ and~$\last_\Rbest$ was a proper ancestor of~$\worst_\Rbest$, then (by definition of~$\worst_\Rbest$), $(\last_q,\pos_q)$ would be an $\eps$-good arc and~$\pos_q$ would wake up an ancestor of~$\pos_r$.
    This would contradict the fact that~$\best_r$ is the best ancestor of~$r$.
    Hence, $\pos_r$ and~$\pos_\ell$ are in the subtree rooted in~$\worst_\Rbest$.

    Let us reason by contradiction and assume that~$(\pos_\ell,\worst_\Rbest)$ is an $\eps$-good pair.
    We recursively define a sequence of robots on the path from~$\worst_\Rbest$ to~$\pos_r$ as follows:
    We initialize~$r^1 := r$ and, if $r^i$ is defined and is not woken up by~$\Rbest$, then we define~$r^{i+1}$ as the robot that wakes up~$r^i$.
    Note that~$(\last_{r^i},\pos_{r^i})$ is an $\eps$-bad return arc, as~$\pos_{r^i}$ is on the path between $\pos_r$ and its best ancestor $\best_r$.
    We now prove by induction on $i$ that performing a trajectory swap consecutively at each~$r^j$ with~$1\leq j < i$ leads to a trajectory labeling that has the same number of $\eps$-bad arcs and the same $\eps$-cost as the initial trajectory labeling~$\lambda$, but contains the $\eps$-bad return arcs~$\{(\pos_\ell,\pos_{r^i})\}\cup \{(\last_{r^{j+1}},\pos_{r^j})\mid 1\leq j < i\}$ instead of the return arcs~$\{(\pos_\ell,\pos_r)\}\cup \{(\last_{r^{j+1}},\pos_{r^{j+1}})\mid 1\leq j < i\}$.
    The base case of~$i=1$ is immediate since~$r^1 = r$.
    Let us now suppose that the statement holds for~$i$ and prove it also holds for~$i+1$ under the assumption that~$r^{i+1}$ exists.
    By induction hypothesis, there exists a trajectory labeling~$\lambda^i$ that contains the $\eps$-bad return arcs~$\{(\pos_\ell,\pos_{r^i})\}\cup \{(\last_{r^{j+1}},\pos_{r^j})\mid 1\leq j < i\}$ instead of the return arcs~$\{(\pos_\ell,\pos_r)\}\cup \{(\last_{r^{j+1}},\pos_{r^{j+1}})\mid 1\leq j < i\}$.
    Note that~$r^{i+1}$ still wakes up~$r^i$.
    So, by performing a trajectory swap at~$\pos_{r^i}$, we switch the return arcs of both~$r^i$ and~$r^{i+1}$.
    That is, the resulting trajectory labeling~$\lambda^{i+1}$ now contains the return arcs~$(\pos_\ell,\pos_{r^{i+1}})$ and~$(\last_{r^{i+1}},\pos_{r^i})$.
    These two return arcs must be $\eps$-bad, as the return arcs of both~$r^{i}$ and~$r^{i+1}$ under~$\lambda^i$ are $\eps$-bad by induction hypothesis and we would otherwise obtain a trajectory labeling that has strictly fewer $\eps$-bad return arcs than the our initial trajectory labeling~$\lambda$.
    This proves the statement for~$i+1$.

    Let~$\lambda'$ be the trajectory labeling having the same number of $\eps$-bad arcs and the same $\eps$-cost as the initial trajectory labeling~$\lambda$, and contains the $\eps$-bad return arcs~$\{(\pos_\ell,\pos_{r^k})\}\cup \{(\last_{r^{j+1}}, \pos_{r^j})\mid 1\leq j < k\}$ instead of the $\eps$-bad return arcs~$\{(\pos_\ell,\pos_r)\}\cup \{(\last_{r^{j+1}},\pos_{r^{j+1}})\mid 1\leq j < k\}$, where~$r^k$ is woken up by~$\Rbest$.

    If~$r^k = \Rworst$, we obtain a contradiction to the assumption that~$(\pos_\ell, \worst_\Rbest)$ is an $\eps$-good pair.
    We therefore assume in the following that~$r^k \neq \Rworst$.
    This implies that~$\worst_\Rbest$ is a proper ancestor of~$\pos_{r^k}$ and~$\pos_{r^k}$ is a proper ancestor of~$\last_{\Rbest}$.
    We now perform two trajectory swaps: first on~$\pos_{r^k}$ and then on~$\worst_\Rbest$,
    and let~$\lambda''$ be the resulting trajectory labeling.
    It holds that~$\lambda''$ now contains the return arcs~$(\pos_\ell,\worst_\Rbest)$, $(\last_{\Rworst},\best_r)$ and~$(\last_\Rbest,\pos_{r^k})$ instead of the return arcs~$(\pos_\ell,\pos_{r^k})$, $(\last_{\Rworst},\worst_\Rbest)$ and~$(\last_\Rbest,\best_r)$ from~$\lambda'$.
    Recall~$(\pos_\ell,\pos_{r^k})$ and $(\last_{\Rworst},\worst_\Rbest)$ are $\eps$-bad return arcs and that~$(\last_\Rbest,\best_r)$ is an $\eps$-good return arc.
    As we assume that~$(\pos_\ell,\worst_\Rbest)$ is an $\eps$-good pair, this implies that both~$(\last_{\Rworst},\best_r)$ and~$(\last_\Rbest,\pos_{r^k})$ must be $\eps$-bad return arcs, as the initial trajectory labeling~$\lambda$ minimized the number of $\eps$-bad return arcs.
    This implies that the number of $\eps$-bad return arcs of~$\lambda$ and~$\lambda''$ are identical and that~$(\last_{\Rworst},\best_r)$ is an $\eps$-bad return arc.
    However, this implies that the $\eps$-bad return arc profile of~$\lambda''$ is lexicographically larger than the one of~$\lambda$, thus contradicting the optimality of~$\lambda$.
    As all cases lead to a contradiction, we conclude that~$(\pos_\ell,\worst_\Rbest)$ is an $\eps$-bad pair, which implies that~$|\worst_\Rbest \pos_{\ell}| > 2-\eps$.
    This completes the proof.
\end{proof}

\begin{restatable}{claim}{lemCloseRobots}%
    \label{claim:close robots}
    Let~$0 < \eps < 1/2$ and let~$P$ be a point in a unit disk~$D$.
    Then, the set of all points in~$D$ with distance at least~$2-\eps$ from~$P$ are contained in a part of the crown with angle at most~$\alpha = 4 \arccos(1-\eps/2)$ and width $\eps$.
\end{restatable}

\begin{proof}
    To begin, remark that any pair of points $\{P, P'\} \in D$  are necessary within a crown of width $\epsilon$ and angle $2\pi$.
    Let $Q$ be a point at distance $1$ from $O$ and such that $P$ lies on segment $OQ$.
    Let $D_Q$ (resp. $D_P$) be the disk of radius $2-\eps$  centered at $Q$ (resp $P$).
    Points $P'$ at distance larger than $2-\eps$ from $P$  in $D$ are within region $D \setminus D_P \subseteq D \setminus D_Q$.

    Let $Q_1$ and $Q_2$ be the two intersection points between $D_Q$ and the unit circle centered at $O$.
    Let $\beta$ be the measure of the angle $\widehat{Q_1QQ_2}$.
    Let $M$ be the middle point of $Q$ and $Q_1$. Since $|OQ|=|OQ_1|=1$, $OQQ_1$ is a isosceles triangle and $MQO$ is a right triangle at $M$.
    Thus, we have $\cos(\beta/2)=|QM|/|QO|=1-\eps/2$, and therefore $\beta = 2 \arccos(1-\eps/2)$.

    Since the inscribed angle theorem states that an angle $\beta$ inscribed in a circle is half of the central angle $\alpha=2\beta$ that intercepts the same arc on the circle, we get the final result.
\end{proof}

\lemSmallCone*
\begin{proof}
    According to \cref{lem:rb on crown}, given an $\eps$-bad return arc $(\last_r,r)$, we know that any robot $x$ of $\Lset_r$ is within a crown with width $\eps$ and at a distance of at most $\eps$ from a leaf $\ell$ of $\wutree_{\bad_r}$.
    Thus, robots of $\Lset_r$ are located in a sub-crown of width $\eps$ and an angle to determine.
    Let us upper bound $\alpha_\ell$ as the maximal angular distance between $\pos_x$ and $\pos_\ell$.
    In the worst case, $|\pos_\ell \pos_x|=\eps$ and $[\pos_\ell \pos_x]$ is a chord of a circle centered at the origin but with radius $1-\eps$.
    Thus we have $\eps=(1-\eps)\chord(\alpha_\ell)=2(1-\eps)\sin(\alpha_\ell/2)$ and $\alpha_\ell=2 \arcsin{\left(\frac{\eps}{2(1-\eps)}\right)}$.

    According to \cref{lem:best ancestor}, $\pos_\ell$ and $\last_\Rbest$ are at distance larger than $2-\eps$ from $\worst_\Rbest$.
    \cref{claim:close robots} implies that $\pos_\ell$ and $\last_\Rbest$ are both within a crown of angle $\alpha=4 \arccos(1-\eps/2)$.
    Positions of robots in $\Lset_r$ are at an angular distance less than $\alpha_\ell$ from $\pos_\ell$.
    For $\eps <0.5, \alpha_\ell < \alpha/2$ and thus, $\pos_\ell,\last_\Rbest$ and positions of robots in $\bigcup_{r\in \Bad_\Rbest} \Lset_r$ are within a crown of angle $\alpha+\alpha_\ell$.
\end{proof}



\section{NP-Hardness}%
\label{app:NP-hardness}

\NPhard*


\proofHardness

\NPcorollary*
\NPmetricproof



\section{Optimal Algorithms}%
\label{app:optimal-algo}

\optimalalgo*


\begin{proof}
    \def\T{\mathsf{T}} %
    We assume we are given a function $d$ that gives a non-negative
    distance between any two robot's initial positions, and let $\P$ be
    the multi-set of robots initial positions. For simplicity, we assume
    that this multi-set is given as a list
    $\P = (\pos_0, \pos_1, \dots, \pos_n)$, where $\pos_0$ is the
    initial position of the source (the awake robot), and where
    $\P^* = (\pos_1, \dots, \pos_n)$ is the list of initial positions of
    all the sleeping robots. For convenience, we extend set notations to
    lists, like for instance $\P^* = \P\setminus\set{\pos_0}$.

    For the sake of presentation, we assume the three algorithms return the value of the optimal makespan only. It is easy to adapt the algorithms so that they also return the corresponding solutions encoded by a wake-up tree (possibly with return arcs), without asymptotically increasing the time and space complexity.

    \subparagraph{\FTP.}

    We consider the variable $T_b(u,X)$ denoting the optimal makespan to wake-up all sleeping robots with initial positions in $X \subseteq \P^*$ from $b \in\set{1,2}$ awake robots at position $u \in \P\setminus X$. Clearly, the wanted optimal makespan for \FTP is given by $T_1(\pos_0,\P^*)$ by definition. All values of $T_b(u,X)$ can be recursively computed as follows.

    We distinguish three cases:

    \begin{enumerate}

        \item If $X = \emptyset$, there is nothing to do for the $b$ robots at $u$. So, we have $T_1(u,\emptyset) = 0$, for any $b\in\set{1,2}$.

        \item If $b=1$, there is only one awake robot at position
              $u$. Then, in order to wake-up all sleeping robots in $X \neq \emptyset$, this robot has to select and go to some position $v\in X$. From that position $v$, two awake robots are available to wake-up all remaining sleeping robots initially located at positions of $X\setminus\set{v}$, and the optimal makespan for this subtask is $T_2(v,X\setminus\set{v})$ by definition. So, the best solution has to select the best $v\in X$, and thus $T_1(u,X) = \min_{v\in X} \set{d(u,v) + T_2(v,X\setminus\set{v})}$.

        \item If $b=2$, there are two awake robots available at position
              $u$. Then, in the optimal solution, one of the two robots at $u$ has to wake-up the robots sleeping on a certain subset $A$ of positions in $X$, whereas the other robot at~$u$ has to wake-up the remaining robots sleeping on positions in $X\setminus A$. Then, the makespan is the maximum of this two branches, and thus $T_2(u,X) = \min_{A\subseteq X} \set{\max\set{ T_1(u,A), T_1(u,X\setminus A) }}$.

    \end{enumerate}

    To summarize, $\forall X\subseteq \P^*$, $\forall u \in \P\setminus X$, $\forall b\in\set{1,2}$:
    \[
        T_b(u,X) ~=~ \left\{%
        \begin{array}{ll}%
            0                                                                                & \mbox{if $X = \emptyset,$} \\[1.2ex]
            \displaystyle\min_{v\in X} \set{d(u,v) + T_2(v,X\setminus\set{v})}               & \mbox{if $b=1,$}           \\[1.5ex]
            \displaystyle\min_{A\subseteq X} \set{\max\set{ T_1(u,A), T_1(u,X\setminus A) }} & \mbox{if $b=2.$}
        \end{array}
        \right.
    \]

    These values can be computing using two tables, say
    $\mathsf{T_1[u][X]}$ and $\mathsf{T_2[u][X]}$ in charge of storing
    $T_1(u,X)$ and $T_2(u,X)$ respectively. These tables are 2D arrays
    composed of $(n+1) \times 2^n$ cells, each one containing some
    distance value. The space complexity is $O(2^n \cdot n)$ as
    claimed. Note that not all the cells of these tables are used (those
    corresponding to $u \in X$ for instance, which is roughly half the
    number of cells, are never used).

    We encode $X$ by an integer $\mathsf{X} \in [0,2^n)$, whose binary
    representation indicates whether each $\pos_i \in \P^*$ is present or
    not in $X$, whereas $u$ is encoded by its index $\mathsf{u}
        \in\range{0}{n}$ of the initial position of $u\in\P$, i.e., such that
    $u = \pos_{\mathsf{u}}$. By scanning first the columns in increasing
    order, and then each row $\mathsf{u}$ of each column $\mathsf{X}$, one can alternatively fill $\T_1$ and $\T_2$. This is valid, because to fill $\T_1$ or $\T_2$, we only need to consider subsets of $X$, whose integer representation is strictly smaller than $\mathsf{X}$ (its binary representation has at least one bit set less). Therefore, all the cells of these tables can be completed by this way.

    Computing $\mathsf{T_1[u][X]}$ requires $|X|\le n$ accesses to table $\T_2$ (for computing the best $v\in X$), whereas $\mathsf{T_2[u][X]}$ requires $2\times 2^{|X|}$ accesses to table $\T_1$ (by enumerating all $A \subseteq X$). Note that each such access requires $O(n)$ time as addressing for these tables are $O(n)$-bit long (in particular for $\mathsf{X}$).

    The number $a_2$ of accesses to $\T_2$, during $\T_1$ calculation, is
    $a_2 \le n \times |\T_1| = O(n^2 \cdot 2^n)$, where
    $|\T_1| = (n+1) \cdot 2^n$ is the number of cells of table $\T_1$.

    To bound the number $a_1$ of accesses to $\T_1$, during $\T_2$
    calculation, we will consider the number $a_1(i)$ of accesses to $\T_1$ when computing $\T_2[u][X]$ where $X$ has a given size $i = |X|$, and sum overall all possible size $i \in\range{0}{n}$, i.e., by setting $a_1 = \sum_{i=0}^n a_1(i)$. In $\T_2$, there are at most $(n+1)\times \binom{n}{i}$ cells to calculate that are concerned with a sub-list $X$ of size $i$ (those corresponding to row $\mathsf{X}$), each one requiring $2\times 2^i$ accesses to $\T_1$. Overall, to complete $\T_2$, we make access to $a_1 = \sum_{i=0}^n a_1(i) \le \sum_{i=0}^n \pare{n\times
            \binom{n}{i}} \cdot \pare{2\times 2^i}$ cells of $\T_1$. Thus
    $a_1 \le 2n \sum_{i=0}^n \binom{n}{i} 2^i = 2n \cdot 3^n$ from
    Newton's binomial formulae.

    In total, the number of accesses to $\T_1$ and $\T_2$ is
    $a_1 + a_2 = O(n \cdot 3^n + n^2\cdot 2^n)$. Therefore, by multiplying
    by the cost of each access, i.e., $O(n)$, the algorithm takes a time
    of $O(3^n \cdot n^2)$ to compute tables $\T_1$ and $\T_2$, and to
    return the optimal makespan $T_1(\pos_0,\P^*) = \T_1[0][2^n-1]$ as
    wanted.

    \subparagraph{\FTRPid.}

    The construction for \FTRPid is very similar to the previous one, and
    can be seen as an extension of \FTRP. As we will see, time and space complexity
    are just scaled by a factor $O(n)$.

    We consider the variable $T_b(u,X,r)$ denoting the optimal makespan to
    wake-up all sleeping robots with initial positions in
    $X \subseteq \P^*$ from $b \in\set{1,2}$ awake robots at position
    $u \in \P\setminus X$, assuming that the initial position of the
    \emph{first} robot is $r$. Here the first robot at $u$ refers to the earliest awake robot at position $u$. If $b=2$, the \emph{second} awake robot at $u$ is the one which has just been awaken at $u$ by the first robot. 

    Clearly, the optimal makespan for \FTRPid is given by $T_1(\pos_0,\P^*,\pos_0)$ by definition. All values of $T_b(u,X,r)$ can be recursively computed as follows.

    As previously, there are three cases:

    \begin{enumerate}

        \item It is clear that $T_b(u,\emptyset,r) = d(u,r)$, for any $b\in\set{1,2}$, since there are no other robots to wake-up. Thus the first robot has to return to its initial position $r$ to complete its task. The possible second awake robot at $u$ has no moves to do in order to return to its initial position.

        \item If $b=1$, only one robot is awake at position $u$, which corresponds to its initial position. Then, in order to complete its task, this robot has to select and go to some position $v\in X \neq \emptyset$. From this position $v$,
              two awake robots have to wake-up all remaining sleeping robots in $X\setminus\set{v}$, the first robot $v$ having still to return to its initial position~$r$. The optimal makespan for this subtask is precisely $T_2(v,X\setminus\set{v},r)$. It follows that $T_1(u,X,r) = \min_{v\in X} \set{d(u,v) + T_2(v,X\setminus\set{v},r)}$.

        \item If $b=2$, two awake robots are available at position
              $u$. Then, in the optimal solution, the first one at $u$ (whose
              return position is $r$) has to wake-up some subset
              $A\subseteq X$ of sleeping robots. Subset $A$ is possibly empty,
              allowing the robot to return directly to $r$ thanks to the
              previous case $X=\emptyset$. The second robot at $u$ has to
              wake-up all the other sleeping robots, those whose initial positions are in $A\setminus X$,
              and has to return to its initial position $u$. Then, the
              makespan is the maximum of these two branches, and thus
              $T_2(u,X,r) = \min_{A\subseteq X} \set{\max\set{ T_1(u,A,r),
                          T_1(u,X\setminus A,u)}}$.

    \end{enumerate}

    To summarize, $\forall X\subseteq \P^*$, $\forall u,r \in \P\setminus X$, $\forall b\in\set{1,2}$:
    \[
        T_b(u,X,r) ~=~ \left\{%
        \begin{array}{ll}%
            d(u,r)                                                                               & \mbox{if $X = \emptyset,$} \\[1.2ex]
            \displaystyle\min_{v\in X} \set{d(u,v) + T_2(v,X\setminus\set{v},r)}                 & \mbox{if $b=1,$}           \\[1.5ex]
            \displaystyle\min_{A\subseteq X} \set{\max\set{ T_1(u,A,r), T_1(u,X\setminus A,u) }} & \mbox{if $b=2.$}
        \end{array}
        \right.
    \]

    Similarly to \FTP, these values can be computing using two tables,
    $\mathsf{T_1[u][X][r]}$ and $\mathsf{T_2[u][X][r]}$ in charge of storing $T_1(u,X,r)$ and $T_2(u,X,r)$ respectively. So, these tables
    are 3D arrays composed of $(n+1) \times 2^n \times (n+1)$ cells, each one containing some distance value. Once completed, the optimal makespan is given by $T_1(\pos_0,\P^*,\pos_0) = \T_1[0][2^n-1][0]$ as wanted.

    Such tables can be completed in the same way as for \FTP, the only difference is that tables for \FTRPid have an extra dimension for $r \in\range{0}{n}$. So, its time and space complexities are at most $n+1$ times the one of \FTP, that is $O(2^n \cdot n^2)$ for the space and $O(3^n \cdot n^3)$ for the time.

    \subparagraph{\FTRPa.}

    The construction for \FTRPa can be seen as a generalization of \FTRPid and \FTP, where the variable has to manage a list of return positions. This list is empty for \FTP. 

    More precisely, we consider the variable $T_b(u,X,R)$ denoting the
    optimal makespan to wake-up all sleeping robots with initial positions
    in $X\subseteq \P^*$ from $b\in\set{1,2}$ awake robots at position
    $u\in \P\setminus X$, assuming that the $b$ awake robots at $u$ and
    all the sleeping robots in $X$ will return to a unique position of $R
        \subseteq \P$, subject to $|R| = b + |X|$.

    Clearly, the optimal makespan for \FTRPa is given by $T_1(\pos_0,\P^*,\P)$ by definition. All values of $T_b(u,X,R)$ can be recursively computed as follows.

    Again, we distinguish three cases:

    \begin{enumerate}

        \item If $|X| = \emptyset$, then we have $|R| = b$ by assumption. It is clear in that case that the $b$ awake robots at $u$ have to return to a unique position of $R$, since there are no other robots to
              wake-up. So, $T_b(u,\emptyset,R) = \max_{r\in R} d(u,r)$, since
              one of the $b$ return arcs must have this length. 

        \item If $b=1$, only one robot is awake at position $u$. Then, in
              order to complete its task, this robot has to select and go to
              some position $v\in X \neq \emptyset$. From this position $v$,
              two awake robots have to wake-up all remaining sleeping robots
              in $X\setminus\set{v}$, and also return to some unique position
              in $R$, including the initial awake robot at $u$. The optimal
              makespan for this subtask is precisely
              $T_2(v,X\setminus\set{v},R)$. It follows that
              $T_1(u,X,R) = \min_{v\in X} \set{d(u,v) +
                      T_2(v,X\setminus\set{v},R)}$.

        \item If $b=2$, two awake robots are available at position
              $u$. Then, in the optimal solution, the first one at $u$ has to
              wake-up some subset $A\subseteq X$ of sleeping robots, and to
              select in $R$ a return position, say $r_1$. Once awake, these
              robots whose initial positions were in $A$ will have to select a unique position in
              $R\setminus\set{r_1}$, forming some subset $S_1$ of positions. Let
              $S = S_1 \cup \set{r_1}$, noting that $S\subseteq R$, and
              $|S| = |A|+1$. Given $A$ and $S$, the optimal time for the first
              robot to wake-up all sleeping robots in $A$ with return
              positions in $S$, given that $|S| = |A|+1$, is $T_1(u,A,S)$ by
              definition.

              At the same time, the second robot at $u$ has to wake-up all remaining sleeping robots of $X$, that is those with initial positions are in $X\setminus A$, and to
              select a return position $r_2 \in R\setminus S$. Once awake, all robots with initial positions in $X\setminus A$ have to select a unique position in $R\setminus (S\cup\set{r_2})$. As $|S| = |A|+1$ and $|R| = |X| + 2$, we have $|R \setminus (S\cup\set{r_2})| = |R| - (|A| + 2) = |X| - |A| = |X\setminus A|$, and so there are the right number of return positions for the sleeping robots in $X\setminus A$. Adding back $r_2$ to the return set $R \setminus (S\cup\set{r_2})$, we get that $R\setminus S$ is the return set of positions for the second robot and all robots with initial positions in $X\setminus A$.
              In particular, $|R\setminus S| = |X\setminus A| + 1$. The optimal time for the second robot to wake-up all sleeping robots in $X\setminus A$ with return positions in $R\setminus S$, given that
              $|R\setminus S| = |X\setminus A| + 1$, is
              $T_1(u,X\setminus A,R\setminus S)$ by definition.

              Therefore, by looking for the best subsets $A$ and $S$, the optimal makespan is the maximum of these two wake-up subtasks, and thus $T_2(u,X,R) = \min_{A\subseteq X} \min_{S\subseteq R} \max{\set{ T_1(u,A,S), T_1(u,X\setminus A,R\setminus S)}}$ with $|S| = |A| + 1$.

    \end{enumerate}

    To summarize, $\forall X\subseteq \P^*$, $\forall u \in \P\setminus X$, $\forall b\in\set{1,2}$, $\forall R\subseteq \P$ such that $|R| = |X| + b$:
    \[
        T_b(u,X,R) ~=~ \left\{%
        \begin{array}{ll}%
            \displaystyle\max_{r\in R} d(u,r)                                    & \mbox{if $X = \emptyset;$} \\[1.5ex]
            \displaystyle\min_{v\in X} \set{d(u,v) + T_2(v,X\setminus\set{v},R)} & \mbox{if $b=1;$}           \\[1.5ex]
            \displaystyle\min_{{\substack{(A,S)\subseteq X\times R                                            \\|S|=|A|+1}}} \set{\max\set{ T_1(u,A,S), T_1(u,X\setminus A,R\setminus S) }} & \mbox{if $b=2.$}
        \end{array}
        \right.
    \]
    Not surprisingly, $T_b(u,X,\emptyset)$ follows that same equations than the ones for solving \FTP.

    As previously, all these values can be computed using two tables,
    $\mathsf{T_1[u][X][R]}$ and $\mathsf{T_2[u][X][R]}$ in charge of
    storing $T_1(u,X,R)$ and $T_2(u,X,R)$ respectively. Notice that
    $R \subseteq \P$ is represented by an integer of $[0,2^{n+1})$, the
    binary representation of its bitmap on $n+1$ bits. These tables are 3D
    arrays composed of $(n+1) \times 2^n \times 2^{n+1}$ cells.
    Thus, the space complexity is $O(4^n \cdot n)$ as claimed. We will see
    later how to save a $\sqrt{n}$ factor on the space complexity with a
    joint representation of $(X,R)$.

    The computation of the tables can be done by scanning first each pair
    $(X,R)$, and then each $u \in \P\setminus X$, and by alternating table
    $\T_1$ and $\T_2$. Once $\T_1$ is completed, the optimal makespan is
    given by $T_1(\pos_0,\P^*,\P) = \T_1[0][2^n-1][2^{n+1}-1]$ as wanted.

    The time complexity can be analyzed by considering the number $a_2$ of
    accesses to $\T_2$, during $\T_1$ calculation, and the number $a_1$ of
    accesses to $\T_1$, during $\T_2$ calculation.

    Computing $\mathsf{T_1[u][X][R]}$ requires $|X| \le n$ accesses to
    $\T_2$. Thus $a_2 \le n \times |\T_1| = O(n^2 \cdot 4^n)$.

    In order to estimate $a_1$, consider the number $N(p,q,k)$ of couples $(A,B)$ where $A \subseteq \range{1}{p}$, $B\subseteq \range{1}{q}$, and such that $|B| = |A| + k$. Using Vandermonde identities, we have (for convenience, we set $\binom{p}{i} = 0$ if $i>p$):
    \[
        N(p,q,k) ~=~ \sum_{i\ge 0} \binom{p}{i} \binom{q}{i+k} =
        \binom{p+q}{p+k} ~.
    \]

    Computing $\mathsf{T_2[u][X][R]}$ requires $2\times N(|X|,|R|,1)$
    accesses to $\T_1$, as we seek for all $(A,S) \subseteq X\times R$
    with $|S| = |A| + 1$. Since, when computing $\mathsf{T_2[u][X][R]}$,
    $b = 2$, we have $|R| = |X| + 2$, and thus
    $2\times N(|X|,|R|,1) = 2\times \binom{2|X|+2}{|X|+1} = O(4^{|X|} /
        \sqrt{|X|}) = O(4^i / \sqrt{i})$ for $i = |X|$.

    In $\T_2$, the number of cells that are concerned with a dimension $X$
    of size $i$ is $(n+1)\times \binom{n}{i} \times \binom{n+1}{i+2}$,
    since we have to consider all $u \in\P$, all $X\subseteq \P^*$ of size
    $i$, and all $R \subseteq \P$ of size $i+2$.

    It follows that the total number of accesses to $\T_1$, when computing
    $\T_2$, is:
    \begin{eqnarray*}
        a_1 &\le& \sum_{i=0}^n \pare{(n+1)\times \binom{n}{i} \times
            \binom{n+1}{i+2}} \cdot O\pare{\frac{4^i}{\sqrt{i}}} \\
        &=& O\pare{ n \sum_{i=0}^n \binom{n}{i} \binom{n+1}{i+2} \cdot \frac{4^i}{\sqrt{i}} } ~=~ O\pare{n \cdot S_n}
    \end{eqnarray*}
    where $S_n$ is the main sum. Using Pascal's formulae, we have $\binom{n}{i} \le \binom{n+2}{i+2}$ and $\binom{n+1}{i+2} \le \binom{n+2}{i+2}$, so that
    \begin{eqnarray*}
        S_n &=& \sum_{i=0}^n \binom{n}{i} \binom{n+1}{i+2} \pare{\frac{2^i}{\sqrt[4]{i}}}^2 ~\le~ \sum_{i=0}^n \binom{n+2}{i+2}^2 \cdot \pare{\frac{2^i}{\sqrt[4]{i}}}^2 \\
        &\le& \sum_{i=0}^n \pare{ \binom{n+2}{i+2}  \cdot \frac{2^i}{\sqrt[4]{i}} }^2 ~\le~ \pare{ \sum_{i=0}^n \binom{n+2}{i+2} \cdot  \frac{2^i}{\sqrt[4]{i}} }^2
    \end{eqnarray*}
    As the maximum term of the sum is reached when $i$ is a constant fraction of $n$, and because $i \mapsto 1/\sqrt[4]{i}$ is convexe, we can bound $\sum_i \binom{n+2}{i+2} \cdot 2^i/\sqrt[4]{i} = O(1/\sqrt[4]{n}\,) \cdot \sum_i \binom{n+2}{i+2} \cdot 2^i = O(1/\sqrt[4]{n}\,) \cdot 3^n$ by Jensen's inegalities and Newton's binomial formulae. It follows that $S_n = O(9^n / \sqrt{n}\,)$, and $a_1 = O(n \cdot S_n) = O(9^n \cdot \sqrt{n}\,)$.

    In total, the number of accesses to $\T_1$ and $\T_2$ is
    $a_1 + a_2 = O(\sqrt{n} \cdot 9^n + n^2 \cdot 4^n)$. Therefore, by
    multiplying by the cost of each access, i.e., $O(n)$, we obtained an
    algorithm of time complexity $O(9^n \cdot n^{3/2})$ as required.

    We observe that we can save a $\sqrt{n}$ factor on the space
    complexity by optimally encoding the pair $(X,R)$ thanks to the
    \emph{Combinatorial Number System} popularized by Knuth. In this
    system, integers are represented as a sum of binomials. So, instead of
    using two distinct dimensions in the tables for $\mathsf{X}$ and
    $\mathsf{R}$, each of size $2^n$, we can encode the pair $(X,R)$ by a
    unique integer denoted by $\mathsf{XR} \in [0,M)$, where $M$ is the
    number of pairs of subsets in $\range{0}{n}\times \range{0}{n+1}$
    whose size differ by one or two. We have
    $M = N(n,n+1,1) + N(n,n+1,2) = O(4^n/\sqrt{n}\,)$. So, the space
    complexity by storing such a $(n+1) \times M$ table decreases to $O(n \cdot M) = O(4^n \cdot \sqrt{n}\,)$.

    However, for the use of this number system, it requires some extra works to encode and decode a pair $(X,R)$ into or from its compact representation $\mathsf{XR}$. In particular, it requires to pre-compute all binomials up to $2n$. This can be done once in $O(n^2)$ arithmetic operations (additions) on $O(n)$-bit integers, that is less than $O(n^3)$ standard operations. Now, encoding/decoding to access a pair $(X,R)$ from $\mathsf{XR}$ or the reverse, require $O(n)$ arithmetic operations or $O(n^2)$ standard operations. So, the time complexity increases by a $O(n)$ factor.
\end{proof}

\end{document}